\documentclass[journal,twoside,web]{ieeecolor}

\usepackage{generic}
\usepackage{cite}
\usepackage{amsmath}
\usepackage{amssymb}
\usepackage{amsfonts}
\usepackage{ntheorem}
\usepackage{mathtools}
\usepackage{commath}
\usepackage{graphicx}
\usepackage{epsfig}
\usepackage{picture}
\usepackage{algorithm}
\usepackage{algpseudocode}
\usepackage{hyperref}
\usepackage{textcomp}
\usepackage{float}
\usepackage{verbatimbox}
\usepackage{booktabs}
\usepackage{listings}
\usepackage{url}
\usepackage{grffile}
\usepackage{times}
\usepackage[cmintegrals, cmbraces]{newtxmath}
\usepackage{tikz}
\usetikzlibrary{shapes.symbols}
\usetikzlibrary{3d, arrows.meta}
\usepackage{wasysym}

\makeatletter
\renewtheoremstyle{plain}
{\item{\theorem@headerfont ##1\ ##2\theorem@separator}~}
{\item{\theorem@headerfont ##1\ ##2\ (##3)\theorem@separator}~}
\theoremseparator{.}
\makeatother

\newtheorem{assumption}{Assumption}
\newtheorem{corollary}{Corollary}
\newtheorem{definition}{Definition}
\newtheorem{lemma}{Lemma}
\newtheorem{proposition}{Proposition}
\newtheorem{remark}{Remark}
\newtheorem{theorem}{Theorem}

\DeclareMathOperator{\col}{col}
\DeclareMathOperator{\diag}{diag}
\DeclareMathOperator{\bdiag}{bdiag}
\DeclareMathOperator{\rank}{rank}

\newcommand{\N}{{\mathbb{N}}}
\newcommand{\R}{{\mathbb{R}}}
\newcommand{\C}{{\mathbb{C}}}
\newcommand{\Ts}{\tau_{\textup{s}}}

\renewcommand{\QED}{\QEDopen}

\def\BibTeX{{\rm B\kern-.05em{\sc i\kern-.025em b}\kern-.08em
    T\kern-.1667em\lower.7ex\hbox{E}\kern-.125emX}}
\begin{document}
\title{Data-Driven Control of Continuous-Time LTI Systems via Non-Minimal Realizations}
\author{Alessandro Bosso, \IEEEmembership{Member, IEEE}, Marco Borghesi, \IEEEmembership{Student Member, IEEE}, Andrea Iannelli, \IEEEmembership{Member, IEEE}, Giuseppe Notarstefano, \IEEEmembership{Member, IEEE}, and Andrew R. Teel, \IEEEmembership{Fellow, IEEE}
\thanks{The research leading to these results has received funding from the European Union's Horizon Europe research and innovation program under the Marie Sk{\l}odowska-Curie Grant Agreement No. 101104404 - \mbox{IMPACT4Mech}.}
\thanks{A. Bosso, M. Borghesi, and G. Notarstefano are with the Department of Electrical, Electronic, and Information Engineering, University of Bologna, Italy. Email: {\tt\small \{alessandro.bosso, m.borghesi, giuseppe.notarstefano\}@unibo.it}}
\thanks{A. Iannelli is with the Institute for System Theory and Automatic Control, University of Stuttgart, Germany. Email: {\tt\small andrea.iannelli@ist.uni-stuttgart.de}
}
\thanks{A. R. Teel is with the Department of Electrical and Computer Engineering, University of California, Santa Barbara, CA, USA. Email: {\tt\small teel@ece.ucsb.edu}}
}

\maketitle

\begin{abstract}
This article proposes an approach to design output-feedback controllers for unknown continuous-time linear time-invariant systems using only input-output data from a single experiment.
To address the lack of state and derivative measurements, we introduce non-minimal realizations whose states can be observed by filtering the available data.
We first apply this concept to the disturbance-free case, formulating linear matrix inequalities (LMIs) from batches of sampled signals to design a dynamic, filter-based stabilizing controller.
The framework is then extended to the problem of asymptotic tracking and disturbance rejection—in short, output regulation—by incorporating an internal model based on prior knowledge of the disturbance/reference frequencies.
Finally, we discuss tuning strategies for a class of multi-input multi-output systems and illustrate the method via numerical examples.
\end{abstract}

\begin{IEEEkeywords}
Data-driven control, output regulation, linear systems, linear matrix inequalities,  uncertain systems.
\end{IEEEkeywords}

\section{Introduction}
\label{sec:introduction}
\IEEEPARstart{D}{ata-driven} methods are rapidly emerging as a central paradigm in automatic control.
Learning controllers from data is crucial for systems with uncertain or unavailable models, and has become increasingly practical thanks to recent advances in computational power and optimization techniques.
While this shift is now accelerating, its roots trace back several decades.
Some of the earliest and still active research areas that exploit data for control purposes include system identification \cite{ljung1999system} and adaptive control \cite{ioannou2012robust}.
More recently, reinforcement learning \cite{sutton2018reinforcement} has gained widespread attention, lying at the intersection of adaptive and optimal control.

Data-driven control techniques can be classified as indirect or direct, depending on whether they rely on an intermediate identification step.
Among direct methods, a popular recent approach involves using linear matrix inequalities (LMIs) to compute control policies offline from a given dataset.
Below, we provide an overview of key contributions in this area.

\subsubsection*{Direct Data-Driven Control via LMIs}
An influential contribution for discrete-time linear time-invariant (LTI) systems is \cite{de2019formulas}, which formulates LMIs where the model is replaced by a batch of data, leveraging state-space results inspired by Willems et al.'s fundamental lemma \cite{willems2005note}.
Another key work in this setting is \cite{van2020data}, which introduces the framework of data informativity.
Alternative LMIs that account for noisy data rely on tools such as the matrix S-lemma \cite{van2020noisy} or Petersen’s lemma \cite{bisoffi2022data}, where the latter is also applied to bilinear systems.
Most related literature focuses on discrete-time systems.
Recent developments in discrete time are dedicated, e.g., to the linear quadratic regulator problem \cite{dorfler2023certainty}, time-varying systems \cite{nortmann2023direct}, and multi-input multi-output (MIMO) systems \cite{li2024controller, alsalti2025notes}.

In contrast, the theory for continuous-time systems remains less developed.
Extensions of the fundamental lemma \cite{schmitz2024continuous, lopez2024input} have emerged only recently.
From a design perspective, LMIs can still be constructed from batches of sampled data.
These methods have been explored in various settings, e.g., linear \cite{de2019formulas, berberich2021data, bisoffi2022data}, switched \cite{bianchi2025data}, and polynomial systems \cite{guo2021data}, whereas \cite{eising2024sampling} studied the impact of sampling continuous-time signals on informativity.
In \cite{hu2025data}, the authors deal with output regulation, a longstanding problem combining asymptotic disturbance rejection and reference tracking \cite{isidori2003robust, isidori2017lectures}.
Notably, all cited works are limited to the state-feedback case and assume access to state derivatives.
This requirement is impractical due to noise sensitivity and poses major challenges for output-feedback design, which would involve multiple differentiations.
In the state-feedback case, the only derivative-free approaches rely on processing state and input signals via integrals \cite{de2023event}, orthogonal polynomial bases \cite{rapisarda2023orthogonal}, or more general linear functionals \cite{ohta2024sampling}.
To the best of the authors’ knowledge, the only work addressing derivative-free output-feedback control is \cite{bosso2024derivative}, which employs linear filters of the input and output signals in both the state-feedback and single-input single-output (SISO) output-feedback settings.

\subsubsection*{Filters and Non-Minimal Realizations}
The filters in \cite{bosso2024derivative} are not to be confused with signal processing tools. Rather, they define the observer dynamics of a non-minimal realization of the plant.
Non-minimal realizations have a longstanding role in the control literature, and the filters used in \cite{bosso2024derivative} date back to classical schemes for adaptive identification \cite{anderson1974adaptive} and adaptive observer design \cite{kreisselmeier1977adaptive}.
Similar filters have also played a central role in model reference adaptive control \cite{narendra1980stable}, and more recently in output-feedback policy and value iteration methods \cite{rizvi2019reinforcement}.
Their connection with non-minimal realizations is well established \cite[Ch. 4]{narendra1989stable}, although the classical analysis is based on transfer function or transfer matrix arguments.

\subsubsection*{Article Contribution}
This article proposes a data-driven control framework for continuous-time LTI systems based on non-minimal realizations.
Building on the state-space approach of \cite{bosso2024derivative}, we address the stabilization and output regulation problems in a general MIMO setting, without state or derivative measurements.
In particular, we adopt a perspective related to nonlinear Luenberger observers \cite{andrieu2006existence, bernard2022observer}, where the initial approach of \cite{luenberger1964observing} is extended to the nonlinear case by using output filters that reconstruct a higher-dimensional system with contracting error dynamics.
Similarly, our framework lifts the original plant onto a higher-dimensional, input-output equivalent system that is fully user-defined except for the output equation, which contains all uncertainties.
This class of non-minimal realizations is inspired by the canonical parameterizations of internal models in output regulation \cite{isidori2003robust}, and we accordingly refer to them as \emph{canonical non-minimal realizations}.
The contributions of this article are as follows:

\subsubsection*{1)}
We present a direct data-driven algorithm that computes a stabilizing output-feedback controller for MIMO systems from an input-output trajectory collected in a single experiment.
Assuming a canonical non-minimal realization is available, the procedure involves filtering the trajectory and using batches of sampled data to define an LMI akin to those first introduced in \cite{hu2025data}.
From the gains computed with the LMI, we obtain a dynamic, observer-based controller.

\subsubsection*{2)} 
We extend the above method to the case of output regulation, where the plant and the available trajectory are both affected by an unknown disturbance assumed to be a sum of constants and sinusoids with known frequencies.
Our approach integrates the previous filters with an internal model unit and solves an LMI to design an observer-based regulator.
    
\subsubsection*{3)} 
We show how to construct a canonical non-minimal realization for any MIMO system with uniform observability index across all outputs.
Furthermore, we prove that, under informative data, the observability index can be directly estimated from the given input-output trajectory.
The centrality of observability indices for MIMO data-driven control has been recognized and exploited in the discrete-time literature \cite{li2024controller}, \cite{alsalti2025notes}.
Similarly, this article develops notions in continuous time for the special case of uniform observability index.

\subsubsection*{Article Organization}
In Section \ref{sec:problem}, we state the data-driven control problems addressed in this article.
In Section \ref{sec:realization}, we introduce the canonical non-minimal realizations.
Then, Sections \ref{sec:stabilization} and \ref{sec:out_reg} describe the algorithms for data-driven stabilization and output regulation, respectively.
The design of canonical non-minimal realizations is presented in Section~\ref{sec:tuning}.
Finally, Section \ref{sec:simulations} showcases some numerical examples and Section \ref{sec:conclusion} concludes the article.
Some auxiliary tools and the more technical proofs are left in the Appendix.

\subsubsection*{Notation}
We use $\N$, $\R$, and $\C$ to denote the sets of natural, real, and complex numbers.
We denote with $I_j$ the identity matrix of dimension $j$ and with $0_{j \times k}$ the zero matrix of dimension $j \times k$.
Given a symmetric matrix $M = M^\top$, $M \succ 0$ (resp. $M \prec 0$) denotes that it is positive definite (resp. negative definite).
We use $\otimes$ to denote the Kronecker product of matrices.
Finally, for any square matrix $M$, we indicate with $\sigma(M)$ its spectrum, i.e., the set of all its eigenvalues.

\section{Problem Statement}\label{sec:problem}
\subsection{Data-Driven Stabilization}\label{sec:problem1}
Consider a continuous-time linear time-invariant system of the form
\begin{equation}\label{eq:plant}
    \begin{split}
         \dot{x} &= Ax + Bu\\
         y &= Cx,
    \end{split}
\end{equation}
where $x \in \R^n$ is the system state, $u \in \R^m$ is the control input, $y \in \R^p$ is the measured output, while $A \in \R^{n \times n}$, $B \in \R^{n \times m}$, and $C \in \R^{p \times n}$ are unknown matrices that satisfy the following assumption, used throughout the article.
\begin{assumption}\label{hyp:ctrb_obs}
    The pair $(A, B)$ is controllable and the pair $(C, A)$ is observable.
\end{assumption}
Suppose that a single experiment is performed on system \eqref{eq:plant} and the resulting input-output data are collected over an interval of length $\tau > 0$:
\begin{equation}\label{eq:dataset}
    (u(t), y(t)), \quad \forall t \in [0, \tau].
\end{equation}
The first problem that we consider is to find a dynamic, output-feedback stabilizing controller for system \eqref{eq:plant} of the form
\begin{equation}\label{eq:ctrl_stabilization}
    \begin{split}
        \dot{\xi} &= A_{\text{c}}\xi + B_{\text{c}}y\\
        u&= C_{\text{c}}\xi + D_{\text{c}}y,
    \end{split}
\end{equation}
where the design of matrices $A_{\text{c}}$, $B_{\text{c}}$, $C_{\text{c}}$, and $D_{\text{c}}$ is solely based on the dataset \eqref{eq:dataset}, without any intermediate identification step.

\subsection{Data-Driven Output Regulation}\label{sec:problem2}
We now consider a more general scenario where we also want to track an output reference and reject a disturbance affecting the plant and the measurements.
This problem, known in the literature as output regulation, is modeled by adding to \eqref{eq:plant} an unknown input $w \in \R^{l}$ as shown below:
\begin{equation}\label{eq:plant_w}
    \begin{split}
        \dot{x} &= Ax + Bu + Pw\\
        y &= Cx + Qw, 
    \end{split}
\end{equation}
where $P \in \R^{n \times l}$ and $Q \in \R^{p \times l}$ are unknown matrices and, as before, $A$, $B$, and $C$ are unknown and satisfy Assumption \ref{hyp:ctrb_obs}.
In system \eqref{eq:plant_w}, the exogenous input $w$ is used to model all reference and disturbance signals that can be obtained from the sum of constant signals and sinusoids with known frequencies but unknown amplitudes and phases.
More precisely, we suppose that $w$ is generated by the autonomous exosystem
\begin{equation}\label{eq:exo}
    \dot{w} = Sw,
\end{equation}
where the matrix $S \in \R^{l \times l}$ is known and neutrally stable, i.e., its minimal polynomial has simple roots on the imaginary axis.
Note that, even if $S$ is known, $w$ is not available as the initial condition $w(0)$ is unknown.

To formulate the tracking objective of output regulation, we suppose that the measured output $y$ in \eqref{eq:plant_w} contains $q \leq p$ components that we want to steer to zero.
In particular, we let:
\begin{equation}\label{eq:CQ}
    C = \begin{bmatrix}
        C_e\\
        C_{\text{r}}
    \end{bmatrix}, \qquad Q = \begin{bmatrix}
        Q_e\\
        Q_{\text{r}}
    \end{bmatrix}
\end{equation}
with $C_e \in \R^{q \times n}$ and $Q_e \in \R^{q \times l}$.
Then, we can split $y$ as
\begin{equation}
    e \coloneqq C_e x + Q_e w, \qquad y_{\text{r}} \coloneqq C_{\text{r}} x + Q_{\text{r}} w,
\end{equation}
where $e \in \R^{q}$ is the \emph{regulated output} (to be steered to zero), while $y_{\text{r}} \in \R^{p - q}$ is the \emph{residual output}, which is used as auxiliary measurement for output-feedback control design.

For the solvability of the problem, we make the following assumption, usually referred to as \emph{non-resonance condition}.
\begin{assumption}\label{hyp:non-resonance}
    It holds that
    \begin{equation}
        \rank\begin{bmatrix}
            A - s I_n & B\\
            C_e & 0
        \end{bmatrix} = n + q,
    \end{equation}
    for all $s \in \sigma(S)$.
\end{assumption}
Loosely speaking, Assumption \ref{hyp:non-resonance} requires that no transmission zero of the triple $(C_e, A, B)$ is also an eigenvalue of $S$.
It is well known that Assumption \ref{hyp:non-resonance} is a necessary condition for the solvability of the output regulation problem \cite[Ch. 4]{isidori2017lectures}.

Suppose that a single experiment is performed on system \eqref{eq:plant_w}, while influenced by exosystem \eqref{eq:exo}, and the resulting input-output data are collected over an interval of length $\tau > 0$:
\begin{equation}\label{eq:dataset_w}
    (u(t), e(t), y_{\text{r}}(t)), \quad \forall t \in [0, \tau].
\end{equation}
Note that, compared with the previous case, the data are affected by the unknown exogenous signal $w(t)$.
Our goal is to find a dynamic, output-feedback controller of the form
\begin{equation}\label{eq:ctrl_w}
    \begin{split}
        \dot{\xi} &= A_{\text{c}}\xi + B_{\text{c}}\begin{bmatrix}
            e \\ y_{\text{r}}
        \end{bmatrix}\\
        u&= C_{\text{c}}\xi + D_{\text{c}}\begin{bmatrix}
            e \\ y_{\text{r}}
        \end{bmatrix},
    \end{split}
\end{equation}
where the matrices $A_{\text{c}}$, $B_{\text{c}}$, $C_{\text{c}}$, and $D_{\text{c}}$ are designed, using only the dataset \eqref{eq:dataset_w} and the prior knowledge of $S$, to ensure the following properties during the online deployment of \eqref{eq:ctrl_w}:
\begin{itemize}
    \item If $w = 0$, the origin $\col(x, \xi) = 0$ is globally exponentially stable for the feedback interconnection of the plant \eqref{eq:plant_w} and the controller \eqref{eq:ctrl_w}.
    \item For any solution $w(t)$ of the exosystem \eqref{eq:exo} and any initial condition of the plant \eqref{eq:plant_w} and the controller \eqref{eq:ctrl_w}, it holds that
    \begin{equation}
        \lim_{t \to \infty}e(t) = 0.
    \end{equation}
\end{itemize}

\section{Canonical Non-Minimal Realizations}\label{sec:realization}
To develop our data-driven control framework, we introduce a non-minimal realization of system \eqref{eq:plant} of the following form:
\begin{equation}\label{eq:non-minimal_plant}
    \begin{split}
        \dot{\zeta} &= (F + LH)\zeta + G u\\
        y &= H\zeta,
    \end{split}
\end{equation}
where $u \in \R^m$ and $y \in \R^p$ are the same of \eqref{eq:plant}, $\zeta \in \R^\mu$, with $\mu \geq n$, is the non-minimal state, $F \in \R^{\mu \times \mu}$, $G \in \R^{\mu \times m}$, and $L \in \R^{\mu \times p}$ are user-defined matrix gains, while $H \in \R^{p \times \mu}$ is an unknown matrix depending on the parameters of \eqref{eq:plant} according to the following novel result, which generalizes \cite[Lem. 3]{bosso2024derivative}.
\begin{lemma}\label{lem:Pi}
    Let Assumption \ref{hyp:ctrb_obs} hold.
    Suppose that, given the plant matrices $A$, $B$, $C$, and the design matrices $F$, $G$, $L$, there exist matrices $\Pi \in \R^{n \times \mu}$ and $H \in \R^{p \times \mu}$ such that
    \begin{equation}\label{eq:Pi}
        \begin{split}
            \Pi(F + LH) &= A\Pi, \quad \Pi G = B\\
            H &= C\Pi.
        \end{split}
    \end{equation}
    Then, $\Pi$ has full-row rank and the controllable and observable subsystem of \eqref{eq:non-minimal_plant} obeys dynamics \eqref{eq:plant}, with state $x = \Pi \zeta$.
\end{lemma}
\begin{proof}
    We first prove that any solution $\Pi$ to \eqref{eq:Pi} satisfies $\rank \Pi = n$.
    By pre-multiplying by $A$ the second equation of \eqref{eq:Pi}, we obtain
    \begin{equation}
        AB = A \Pi G = \Pi(F + LH)G.
    \end{equation}
    Repeat this process and stack the resulting vectors to obtain:
    \begin{equation}\label{eq:Pi_full_rank}
        \Pi M = \begin{bmatrix}
            B & AB & \cdots & A^{n-1}B
        \end{bmatrix},
    \end{equation}
    where:
    \begin{equation}
        M \coloneqq \begin{bmatrix}
            G & (F + LH)G & \cdots & (F + LH)^{n-1}G
        \end{bmatrix}.
    \end{equation}
    Since $(A, B)$ is controllable by Assumption \ref{hyp:ctrb_obs}, from \eqref{eq:Pi_full_rank} we obtain that $\rank \Pi M = n$, which implies that $\rank \Pi = n$.

    We now focus on system \eqref{eq:non-minimal_plant}.
    Define
    \begin{equation}\label{eq:Kalman_obs}
        \begin{bmatrix}
            x_{\bar{o}} \\ x
        \end{bmatrix} \coloneqq \begin{bmatrix}
            \Xi \\ \Pi
        \end{bmatrix}\zeta,
    \end{equation}
    where $\Xi$ contains $\mu - n$ linearly independent rows such that $[\,\Xi^\top \; \Pi^\top]$ is non-singular.
    From \eqref{eq:non-minimal_plant} and \eqref{eq:Pi}, we obtain
    \begin{equation}\label{eq:Pi_zeta_dyn}
        \begin{split}
            \dot{x} &= \Pi(F + LH)\zeta + \Pi G u = Ax + Bu\\
            y&= H\zeta = C\Pi\zeta = Cx,
        \end{split}
    \end{equation}
    and, thus, we achieve the following Kalman decomposition:
    \begin{equation}\label{eq:Kalman_decomposition}
        \begin{split}
            \begin{bmatrix}
                \dot{x}_{\bar{o}} \\ \dot{x}
            \end{bmatrix} &= \begin{bmatrix}
                A_{\bar{o}} & A_{\times}\\
                0_{n \times (\mu - n)} & A
            \end{bmatrix}\begin{bmatrix}
                x_{\bar{o}} \\ x
            \end{bmatrix} + \begin{bmatrix}
                B_{\bar{o}} \\ B
            \end{bmatrix}u\\
            y &= \begin{bmatrix}
                \, 0_{p \times (\mu - n)} \, & C\,
            \end{bmatrix} \begin{bmatrix}
                x_{\bar{o}} \\ x 
            \end{bmatrix},
        \end{split}
    \end{equation}
    for some matrices $A_{\bar{o}}$, $A_{\times}$, $B_{\bar{o}}$.
    The statement follows by noticing that the $x$-subsystem is controllable and observable by Assumption \ref{hyp:ctrb_obs}.
\end{proof}

We postpone to Section \ref{sec:tuning} the design of $F$, $G$, and $L$ to ensure the existence of $\Pi$ and $H$ for a class of MIMO systems.
\begin{remark}\label{rem:transfer}
    Lemma \ref{lem:Pi} ensures that the transfer matrices of \eqref{eq:plant} and \eqref{eq:non-minimal_plant} coincide, i.e., $C(sI_n - A)^{-1} B = H(sI_\mu - F - LH)^{-1}G$.
    This property can be verified by using $(sI_n - A)^{-1} = I_n/s + A/s^2 + A^2/s^3 + \ldots$ and the same series expansion for $(sI_\mu - F - LH)^{-1}$, then applying \eqref{eq:Pi} recursively to each term.
\end{remark}
The next result further characterizes the Kalman decomposition obtained with Lemma \ref{lem:Pi}.
\begin{lemma}\label{lem:eig_F+LH}
    Under the hypotheses of Lemma \ref{lem:Pi}, it holds that
    \begin{equation}
        \sigma(F + LH) \subset (\sigma(A)\cup\sigma(F)).
    \end{equation}
    In particular, the unobservable eigenvalues of $F + LH$ are eigenvalues of $F$, i.e., $\sigma(A_{\bar{o}}) \subset \sigma(F)$, with $A_{\bar{o}}$ given in \eqref{eq:Kalman_decomposition}.
\end{lemma}
\begin{proof}
    From the decomposition \eqref{eq:Kalman_decomposition}, it holds that $\sigma(F + LH) = \sigma(A_{\bar{o}}) \cup \sigma(A)$.
    By \cite[Ch. 3, Cor. 4.6]{antsaklis1997linear}, any eigenvalue $s \in \sigma(A_{\bar{o}})$ is such that there exists a vector $v \in \C^\mu$, with $v \neq 0$, satisfying:
    \begin{equation}
        (F + LH)v = s v, \quad Hv = 0.
    \end{equation}
    Using the second equation in the first one, we obtain
    \begin{equation}
        Fv = s v,
    \end{equation}
    which proves the statement.
\end{proof}

Note that system \eqref{eq:non-minimal_plant} can be rewritten as
\begin{equation}\label{eq:non-minimal_plant_filter}
    \begin{split}
        \dot{\zeta} &= F\zeta + Gu + Ly\\
        y &= H\zeta,
    \end{split}
\end{equation}
where the unknown matrix $H$ appears only in the output equation, while the matrices $F$, $G$, and $L$ of the differential equation are user-defined, thus they are known.
If we further require that $F$ is Hurwitz, we finally obtain the following novel definition, which plays a central role in this article.
\begin{definition}\label{def:canonical_realization}
    Under Assumption \ref{hyp:ctrb_obs}, system \eqref{eq:non-minimal_plant} is said to be a \emph{canonical non-minimal realization} of the plant \eqref{eq:plant} if $F$, $G$, $L$, and $H$ are such that:
    \begin{itemize}
        \item There exists $\Pi$ such that $\Pi$ and $H$ solve equation \eqref{eq:Pi}.
        \item $F$ is Hurwitz.
    \end{itemize}
    Furthermore, the realization is said to be \emph{strong} if the pair $(F + LH, G)$ is controllable.
    It is said to be \emph{weak} otherwise.
\end{definition}
The next result is an immediate consequence of Definition \ref{def:canonical_realization}.
\begin{corollary}\label{lem:st_det}
    Under Assumption \ref{hyp:ctrb_obs}, let \eqref{eq:non-minimal_plant} be a canonical non-minimal realization of the plant \eqref{eq:plant}.
    Then, the pair $(F + LH, G)$ is stabilizable and the pair $(H, F + LH)$ is detectable.
\end{corollary}
\begin{proof}
    Note that $x_{\bar{o}}$ in \eqref{eq:Kalman_decomposition} contains all unobservable and uncontrollable states.
    Thus, stabilizability and detectability follow from Lemma \ref{lem:eig_F+LH} with $F$ Hurwitz.
\end{proof}

We now use Definition \ref{def:canonical_realization} to solve the data-driven stabilization and output regulation problems of Section \ref{sec:problem}.

\section{Data-Driven Stabilization}\label{sec:stabilization}

The approach to solve the stabilization problem of Section \ref{sec:problem1} is summarized in Algorithm \ref{alg:stabilization}, where we use the dataset \eqref{eq:dataset} (reported in \eqref{eq:dataset_1} for convenience) to construct the controller \eqref{eq:controller}.
In Fig. \ref{fig:stabilization}, we show the closed-loop interconnection of the plant \eqref{eq:plant} with the controller.
In the following, we illustrate the design of Algorithm \ref{alg:stabilization} and its theoretical guarantees.

Let Assumption \ref{hyp:ctrb_obs} hold, and suppose that $F$, $G$, and $L$ have been designed so that \eqref{eq:non-minimal_plant} is a canonical non-minimal realization of \eqref{eq:plant}, with $H$ and $\Pi$ solving \eqref{eq:Pi}.
Since \eqref{eq:non-minimal_plant} can be written as \eqref{eq:non-minimal_plant_filter}, we introduce a replica of its dynamics:
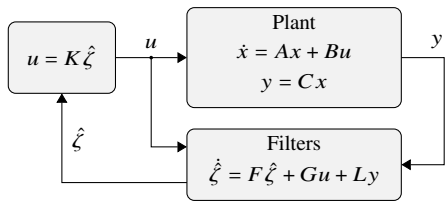
\begin{figure}[b!]

    \vspace{-5pt}
    
    \centering
    \begin{tikzpicture}[scale = 0.95]
		\draw[-{Triangle}] (-1, 0) -- (0, 0);
		\node[anchor = south] at (-0.5, 0) {\footnotesize $u$};
        \draw[-{Triangle}] (-0.5, 0) -- (-0.5, -1.25) -- (0, -1.25);
        \filldraw[fill=black] (-0.5, 0) circle (0.5pt);
		\filldraw[fill=gray!10, rounded corners=1mm](0, -0.7) rectangle (3, 0.7);
        \node[anchor = center] at (1.5, 0.5) {\footnotesize Plant};
        \node[anchor = center] at (1.5, 0.1) {\footnotesize $\dot{x} = Ax + Bu$};
        \node[anchor = center] at (1.5, -0.35) {\footnotesize $y = C x$};
        \draw[-{Triangle}] (3, 0) -- (3.6, 0) -- (3.6, -1.5) -- (3, -1.5);
        \node[anchor = south] at (3.5, 0) {\footnotesize $y$};
        \filldraw[fill=gray!10, rounded corners=1mm](0, -2) rectangle (3, -1);
        \node[anchor = center] at (1.5, -1.2) {\footnotesize Filters};
        \node[anchor = center] at (1.5, -1.6) {\footnotesize $\dot{\hat{\zeta}} = F\hat{\zeta} + G u + Ly$};
        \filldraw[fill=gray!10, rounded corners=1mm](-2.5, -0.5) rectangle (-1, 0.5);
        \node[anchor = center] at (-1.75, 0) {\footnotesize $u = K\hat{\zeta}$}; 
        \draw[-{Triangle}] (0, -1.75) -- (-1.75, -1.75) -- (-1.75, -0.5);
        \node[anchor = west] at (-1.75, -1.15) {\footnotesize $\hat{\zeta}$};
	\end{tikzpicture}
    
    \caption{Implementation of the controller \eqref{eq:controller} designed in Algorithm \ref{alg:stabilization}.}
    \label{fig:stabilization}
    
\end{figure}
\begin{equation}\label{eq:obs}
    \dot{\hat{\zeta}} = F\hat{\zeta} + Gu + Ly,
\end{equation}
which can be seen as an observer of $\zeta$ or, equivalently, a filter of the input-output data.
Note that $\Pi$ and $H$ are not available for design and are only used for analysis.

\begin{algorithm}[t]
    \caption{Data-Driven Stabilization}\label{alg:stabilization}
    \begin{algorithmic}
    \State \hspace{-0.47cm} \textbf{Initialization} 
    \State \emph{Dataset:}
    \begin{equation}\label{eq:dataset_1}
        (u(t), y(t)), \qquad \forall t \in [0, \tau].
    \end{equation}
    \State \emph{Tuning:} $F$, $G$, $L$ such that \eqref{eq:non-minimal_plant} is a canonical non-minimal realization of \eqref{eq:plant}; $(F_0, G_0)$ as in \eqref{eq:Delta}; number of samples $N \in \N$, $N \geq 1$, $\Ts \coloneqq \tau/N$.
    \State \hspace{-0.47cm} \textbf{Data Batches Construction} 
    \State \emph{Filter of the data:}
    \begin{equation}\label{eq:filter}
        \dot{\hat{\zeta}}(t) = F\hat{\zeta}(t) + Gu(t) + Ly(t),\quad \forall t \in [0, \tau].
    \end{equation}
    Initialization: $\hat{\zeta}(0) = 0 \in \R^{\mu}$.
    \State \emph{Auxiliary Dynamics:}
    \begin{equation}\label{eq:aux}
        \dot{\chi}(t) = F_0\chi(t),\quad \forall t \in [0, \tau].
    \end{equation}
    Initialization: $\chi(0) = G_0 \in \R^\delta$.
    \State \emph{Sampled data batches:}
    \begin{equation}\label{eq:batches}
        \begin{split}
            \mathcal{U} &\coloneqq \begin{bmatrix}
                u(0) & u(\Ts) & \cdots & \,u((N - 1)\Ts)
            \end{bmatrix} \in \R^{m \times N}\\
            \mathcal{X} &\coloneqq \begin{bmatrix}
                \chi(0) & \chi(\Ts) & \!\cdots & \!\chi((N-1)\Ts)
            \end{bmatrix} \in \R^{\delta \times N}\\
            \mathcal{Z}& \coloneqq \begin{bmatrix}
                \hat{\zeta}(0) & \hat{\zeta}(\Ts) & \cdots & \hat{\zeta}((N-1)\Ts)
            \end{bmatrix} \in \R^{\mu \times N}\\
            \dot{\mathcal{Z}}& \coloneqq \begin{bmatrix}
                \dot{\hat{\zeta}}(0) & \dot{\hat{\zeta}}(\Ts) & \cdots & \dot{\hat{\zeta}}((N-1)\Ts)
            \end{bmatrix} \in \R^{\mu \times N}
        \end{split}
    \end{equation}
    \State \hspace{-0.47cm} \textbf{Stabilizing Gain Computation}
    \State \emph{LMI:} find $\mathcal{P} \in \R^{\mu \times \mu}$, $\mathcal{Q} \in \R^{N \times \mu}$ such that:
    \begin{equation}\label{eq:LMI}
        \begin{cases}
            \mathcal{P} = \mathcal{P}^\top \succ 0\\
            \dot{\mathcal{Z}}\mathcal{Q} + \mathcal{Q}^\top\dot{\mathcal{Z}}^\top \prec 0\\
            \begin{bmatrix}
                0_{\delta \times \mu}\\
                \mathcal{P}
            \end{bmatrix} = \begin{bmatrix}
                \mathcal{X}\\
                \mathcal{Z}
            \end{bmatrix} \mathcal{Q}.
        \end{cases}
    \end{equation}
    \State \emph{Control gain:}
    \begin{equation}\label{eq:K}
        K = \mathcal{U}\mathcal{Q}\mathcal{P}^{-1}.
    \end{equation}
    \State \hspace{-0.47cm} \textbf{Control Deployment}
    \State \emph{Control law:}
    \begin{equation}\label{eq:controller}
        \dot{\xi} = (F + GK)\xi + Ly, \qquad  u = K\xi.
    \end{equation}
    Initialization: $\xi(0) \in \R^{\mu}$ arbitrary.
    \end{algorithmic}
\end{algorithm}
Define the error
\begin{equation}\label{eq:epsilon}
    \epsilon \coloneqq x - \Pi \hat{\zeta}.
\end{equation}
Using equation \eqref{eq:Pi} and the coordinates \eqref{eq:epsilon}, the interconnection of \eqref{eq:plant} and \eqref{eq:obs} can be written as follows:
\begin{equation}\label{eq:eps_zeta}
    \begin{bmatrix}
        \dot{\epsilon}\\
        \dot{\hat{\zeta}}
    \end{bmatrix} = \begin{bmatrix}
        A - \Pi L C & 0_{n \times \mu}\\
        LC & F + LH
    \end{bmatrix}\begin{bmatrix}
        \epsilon\\
        \hat{\zeta}
    \end{bmatrix} + \begin{bmatrix}
        0_{n \times m} \\ G
    \end{bmatrix} u.
\end{equation}
Note that, since the pair $(F + LH, G)$ is stabilizable by Corollary \ref{lem:st_det}, also system \eqref{eq:eps_zeta} is stabilizable thanks to the next result.
The proof is given in Appendix \ref{sec:proofs}.
\begin{lemma}\label{lem:min_poly}
    In system \eqref{eq:eps_zeta}, the roots of the minimal polynomial of $A - \Pi L C$ are a subset of the roots of the minimal polynomial of $F$.
    As a consequence, $A - \Pi L C$ is Hurwitz.
\end{lemma}

Thanks to Lemma \ref{lem:min_poly} and the cascade structure of \eqref{eq:eps_zeta}, we conclude that the interconnection can be stabilized by a feedback law of the form $u = K\hat{\zeta}$.
In our scenario, $K$ cannot be designed via model-based techniques because $F + LH$, containing the plant parameters, is unknown.
Therefore, in the following, we compute $K$ via data-driven techniques.
\begin{remark}
    The philosophy followed by this article can be regarded as the data-driven equivalent of classical controllers based on the interconnection of an observer and a feedback law using the observer states.
    To see the parallelism with the model-based literature, if we let $\mu = n$, we obtain that \eqref{eq:obs} has the same structure as originally proposed by Luenberger in \cite{luenberger1964observing}.
    If $\Pi = I_n$ (so that $\zeta = x$), \eqref{eq:obs} becomes the observer commonly known in the literature:
    \begin{equation}\label{eq:Luenberger}
        \dot{\hat{x}} = (A - LC)\hat{x} + Bu + Ly,
    \end{equation}
    where $A - LC = F$ is Hurwitz.
    In the data-driven setting, we rely on canonical non-minimal realizations because \eqref{eq:Luenberger} cannot be implemented since $A$, $B$, and $C$ are unknown.
\end{remark}

Given the dataset \eqref{eq:dataset_1}, we filter the input-output signals by simulating \eqref{eq:filter} over $[0, \tau]$.
By looking at \eqref{eq:eps_zeta}, one may wonder if the signals $\hat{\zeta}(t)$ and $\dot{\hat{\zeta}}(t)$ could be used to obtain a data-based representation of the dynamics to be controlled.
The main obstacle to this idea is that $\dot{\hat{\zeta}}(t)$ is affected by $\epsilon(t)$, which is unknown because $x(t)$ and $\Pi$ in \eqref{eq:epsilon} are not available.

However, it turns out that, with some manipulations of \eqref{eq:eps_zeta}, it is possible to replace $\epsilon(t)$ with a user-generated signal $\chi(t)$.
Related ideas can be found in \cite{hu2025data} and \cite{bosso2024derivative}.
To show this notable feature, let
\begin{equation}
    m_F(s) \coloneqq s^\delta + \theta_{f, \delta-1}s^{\delta - 1} + \theta_{f, 1}s + \theta_{f, 0}
\end{equation}
be the minimal polynomial of $F$, with degree $\delta$.
Then, define
\begin{equation}\label{eq:Delta}
    F_0 \! \coloneqq \! \begin{bmatrix}
        0 & 1 & \cdots & 0\\
        \vdots & & \ddots & \vdots\\
        0 & 0 & \cdots & 1\\
        -\theta_{f, 0} & -\theta_{f,1} & \cdots & -\theta_{f, \delta-1}
    \end{bmatrix}\!, \qquad G_0 \! \coloneqq \! \begin{bmatrix}
        0\\
        \vdots\\
        0\\
        \omega_f
    \end{bmatrix}\!,
\end{equation}
with $\omega_f \neq 0$ arbitrary.
We obtain the following property, whose proof is given in Appendix \ref{sec:proofs}.
\begin{lemma}\label{lem:chi}
    Consider system \eqref{eq:eps_zeta}.
    There exists a matrix $D \in \R^{\mu \times \delta}$ such that $LC\epsilon(t) = D\chi(t)$ for all $t \in [0, \tau]$, where $\epsilon(t)$ is the solution of $\dot{\epsilon}(t) = (A - \Pi L C)\epsilon(t)$ with $\epsilon(0) = x(0)$, and $\chi(t)$ is the solution of $\dot{\chi}(t) = F_0\chi(t)$, with $\chi(0) = G_0$.
\end{lemma}

Given Lemma \ref{lem:chi} and system \eqref{eq:eps_zeta}, we can combine the simulations of \eqref{eq:filter} and \eqref{eq:aux} to obtain that the dataset \eqref{eq:dataset_1} and the simulated trajectories satisfy, for $t \in [0, \tau]$,
\begin{equation}\label{eq:chi_zeta}
    \begin{bmatrix}
        \dot{\chi}(t)\\
        \dot{\hat{\zeta}}(t)
    \end{bmatrix}\! = \!\begin{bmatrix}
        F_0\! & 0_{\delta \times \mu}\\
        D\! & F + LH
    \end{bmatrix}\!\begin{bmatrix}
        \chi(t)\\
        \hat{\zeta}(t)
    \end{bmatrix} + \begin{bmatrix}
        0_{\delta \times m} \\ G
    \end{bmatrix}\! u(t), \;\; \begin{bmatrix}
        \chi(0)\\
        \hat{\zeta}(0)
    \end{bmatrix} \! = \! \begin{bmatrix}
        G_0 \\ 0
    \end{bmatrix}
\end{equation}
where the signals $(\chi(t), \hat{\zeta}(t))$, the derivative $\dot{\hat{\zeta}}(t)$, and the input $u(t)$ are available for measurement.

As a consequence, following the data-driven approaches for continuous-time systems \cite{bosso2024derivative}, we construct data batches of the form \eqref{eq:batches} by sampling the continuous-time signals at $N$ distinct time instants.
In Algorithm \ref{alg:stabilization}, we choose to sample with constant sampling time $\Ts$ only for simplicity, and non-uniform sampling could also be employed.
Note that the informativity properties of the batch may be significantly affected by the sampling strategy.

Once the batches \eqref{eq:batches} have been extracted from the continuous-time signals, we can compute the gain $K$ from the linear matrix inequality \eqref{eq:LMI} and equation \eqref{eq:K}, similar to \cite{hu2025data}.
The resulting stabilizing controller to be deployed online is \eqref{eq:controller}, which is the combination of the filters \eqref{eq:obs} and the feedback gain $u = K \hat{\zeta}$.
Again, we remark that this structure can be seen as the interconnection of an observer and a feedback based on the observer states.
Note that we denoted the controller state with $\xi = \hat{\zeta}$ to put \eqref{eq:controller} in the same form of \eqref{eq:ctrl_stabilization}, as described in the problem statement.

The following result provides formal guarantees for the effectiveness of Algorithm \ref{alg:stabilization}.
The proof is based on the arguments found in \cite{de2019formulas, hu2025data}.

\begin{theorem}\label{thm:stabilization}
    Consider Algorithm \ref{alg:stabilization} and let Assumption \ref{hyp:ctrb_obs} hold.
    Then:
    \begin{enumerate}
        \item The LMI \eqref{eq:LMI} is feasible if
        \begin{equation}\label{eq:E}
            \rank\begin{bmatrix}
                \mathcal{X} \\ \mathcal{Z} \\ \mathcal{U}
            \end{bmatrix} = 
            \delta + \mu + m.
        \end{equation}
        \item For any solution $\mathcal{P}$, $\mathcal{Q}$ of \eqref{eq:LMI}, the gain $K$ computed from \eqref{eq:K} is such that  $F + LH + GK$ is Hurwitz.
        Thus, the controller \eqref{eq:controller} solves the data-driven stabilization problem of Section \ref{sec:problem1}.
    \end{enumerate}
\end{theorem}
\begin{proof}
    \emph{1):} Since the pair $(F + LH, G)$ is stabilizable by Corollary \ref{lem:st_det}, there exist matrices $\mathcal{P}$ and $K$ such that
    \begin{equation}\label{eq:PK_LMI}
        \begin{cases}
            \mathcal{P} = \mathcal{P}^\top \succ 0\\
            (F + LH + GK) \mathcal{P} + \mathcal{P}^\top (F + LH + GK)^\top \prec 0.
        \end{cases}
    \end{equation}
    Given any $\mathcal{P}$ and $K$ satisfying \eqref{eq:PK_LMI}, condition \eqref{eq:E} implies that there exists a matrix $\mathcal{M}$ such that \cite[Thm. 2]{de2019formulas}, \cite[Thm. 2]{hu2025data}:
    \begin{equation}\label{eq:M}
        \begin{bmatrix}
            0_{\delta \times \mu}\\I_\mu \\ K
        \end{bmatrix} = \begin{bmatrix}
            \mathcal{X} \\ \mathcal{Z} \\ \mathcal{U}
        \end{bmatrix}\mathcal{M}.
    \end{equation}
    Notice that, using \eqref{eq:chi_zeta}, we can write
    \begin{equation}\label{dZ}
        \dot{\mathcal{Z}} = D\mathcal{X} + (F + LH)\mathcal{Z} + G\mathcal{U}.
    \end{equation}
    Using \eqref{eq:M} and \eqref{dZ}, it holds that:
    \begin{equation}\label{eq:F_LH_GK}
        \begin{split}
            F + LH + GK &= \begin{bmatrix}
                F + LH & G
            \end{bmatrix}\begin{bmatrix}
                \mathcal{Z} \\ \mathcal{U}
            \end{bmatrix}\mathcal{M}\\
            &= (\dot{\mathcal{Z}} - D\mathcal{X})\mathcal{M} = \dot{\mathcal{Z}}\mathcal{M}.
        \end{split}
    \end{equation}
    Then, combining \eqref{eq:PK_LMI} and \eqref{eq:F_LH_GK}, we obtain:
    \begin{equation}\label{eq:Lyapunov}
        \begin{cases}
            \mathcal{P} = \mathcal{P}^\top \succ 0\\
            \dot{\mathcal{Z}}\mathcal{M} \mathcal{P} + \mathcal{P}^\top \mathcal{M}^\top \dot{\mathcal{Z}}^\top \prec 0.
        \end{cases}
    \end{equation}
    Let $\mathcal{Q} = \mathcal{M}\mathcal{P}$ in \eqref{eq:Lyapunov} and post-multiply both sides of \eqref{eq:M} by $\mathcal{P}$ to obtain \eqref{eq:LMI}.

    \emph{2):} Suppose that there exist $\mathcal{P}$ and $\mathcal{Q}$ that satisfy \eqref{eq:LMI}.
    From the third equation of \eqref{eq:LMI}, we obtain that $\mathcal{Z}^\dagger \coloneqq \mathcal{Q}\mathcal{P}^{-1}$ is a right inverse of $\mathcal{Z}$.
    Using \eqref{dZ} and $\mathcal{Z}^\dagger$, it follows that
    \begin{equation}\label{eq:F_LH}
        F + LH = (\dot{\mathcal{Z}} - D\mathcal{X} - G\mathcal{U})\mathcal{Z}^\dagger.
    \end{equation}
    From \eqref{eq:K}, \eqref{eq:F_LH}, and $\mathcal{X}\mathcal{Q} = 0_{\delta \times \mu}$ due to \eqref{eq:LMI}, the following identity holds:
    \begin{equation}\label{eq:Zdot_Z+}
        F + LH + GK = (\dot{\mathcal{Z}} - D\mathcal{X})\mathcal{Z}^{\dagger} = \dot{\mathcal{Z}}\mathcal{Z}^{\dagger}.
    \end{equation}
    Then, using $\mathcal{P} = \mathcal{Z}\mathcal{Q}$ from the equality in \eqref{eq:LMI}, the first inequality of \eqref{eq:LMI} can be rewritten as
    \begin{equation}
        \dot{\mathcal{Z}}\mathcal{Z}^\dagger \mathcal{P} + \mathcal{P}(\mathcal{Z}^{\dagger})^\top\dot{\mathcal{Z}}^\top \prec 0,
    \end{equation}
    which implies from \eqref{eq:Zdot_Z+} that $F + LH + GK$ is Hurwitz.
                
    Finally, to show that the controller \eqref{eq:controller} solves the data-driven stabilization problem, we rewrite the closed-loop interconnection of \eqref{eq:plant} and \eqref{eq:controller} using \eqref{eq:eps_zeta} with $u = K\xi$:
    \begin{equation}\label{eq:cascade_compact}
        \begin{bmatrix}
            \dot{\epsilon}\\
            \dot{\xi}
        \end{bmatrix} = \begin{bmatrix}
            A - \Pi L C & 0_{n \times \mu}\\
            LC & F + LH + GK
        \end{bmatrix}\begin{bmatrix}
            \epsilon\\
            \xi
        \end{bmatrix}.
    \end{equation}
    Since $A - \Pi L C$ and $F + LH + GK$ are Hurwitz, global exponential stability of \eqref{eq:cascade_compact} follows from standard results for cascaded linear systems.
\end{proof}

\begin{remark}
    In this section, we have presented an LMI for the noise-free scenario.
    Future work will be devoted to formulating LMIs that address the case where \eqref{eq:dataset_1} is affected by bounded noise, similar to \cite{van2020noisy} and \cite{bisoffi2022data}.
    Also, it will be worth investigating the impact of the filter gains $F$, $G$, $L$ in case the noise has a bandwidth characterization.
\end{remark}

\section{Data-Driven Output Regulation}\label{sec:out_reg}
The data-driven output regulation problem of Section \ref{sec:problem2} is solved in this section by Algorithm \ref{alg:out_reg}, which 
uses the dataset \eqref{eq:dataset_w} (rewritten in \eqref{eq:dataset_2} for convenience) to design the controller \eqref{eq:controller_out_reg}.
We refer to Figure \ref{fig:out_reg_scheme} for a depiction of the interconnection of the controller with the plant \eqref{eq:plant_w} and the exosystem \eqref{eq:exo}.
Below, we illustrate the derivation of Algorithm \ref{alg:out_reg} and its theoretical guarantees under Assumptions \ref{hyp:ctrb_obs} and \ref{hyp:non-resonance}.

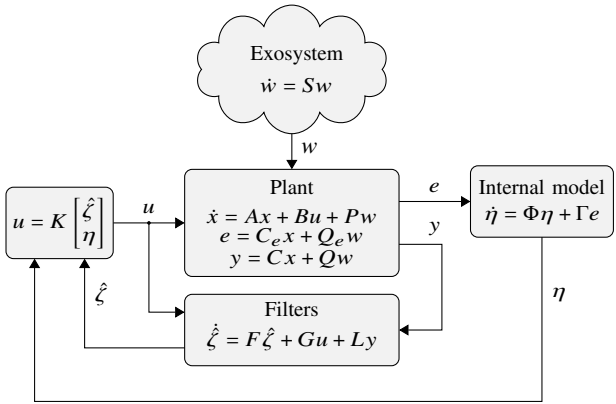
\begin{figure}[b!]

    \vspace{-10pt}
    
    \centering
    \begin{tikzpicture}[scale = 0.95]
		\draw[-{Triangle}] (-1, 0) -- (0, 0);
		\node[anchor = south] at (-0.5, 0) {\footnotesize $u$};
        \draw[-{Triangle}] (-0.5, 0) -- (-0.5, -1.25) -- (0, -1.25);
        \filldraw[fill=black] (-0.5, 0) circle (0.5pt);
		\filldraw[fill=gray!10, rounded corners=1mm](0, -0.75) rectangle (3, 0.75);
        \node[anchor = center] at (1.5, 0.5) {\footnotesize Plant};
        \node[anchor = center] at (1.5, 0.1) {\footnotesize $\dot{x} = Ax + Bu + Pw$};
        \node[anchor = center] at (1.5, -0.2) {\footnotesize $e = C_e x + Q_e w$};
        \node[anchor = center] at (1.5, -0.5) {\footnotesize $y = C x + Q w$};
		\draw[-{Triangle}] (3, 0.3) -- (4, 0.3);
		\node[anchor = south] at (3.5, 0.3) {\footnotesize $e$};
        \draw[-{Triangle}] (3, -0.3) -- (3.6, -0.3) -- (3.6, -1.5) -- (3, -1.5);
        \node[anchor = south] at (3.5, -0.3) {\footnotesize $y$};
        \filldraw[fill=gray!10, rounded corners=1mm](0, -2) rectangle (3, -1);
        \node[anchor = center] at (1.5, -1.2) {\footnotesize Filters};
        \node[anchor = center] at (1.5, -1.6) {\footnotesize $\dot{\hat{\zeta}} = F\hat{\zeta} + G u + Ly$};
        \filldraw[fill=gray!10, rounded corners=1mm](4, -0.2) rectangle (6, 0.8);
        \node[anchor = center] at (5, 0.5) {\footnotesize Internal model};
		\node[anchor = center] at (5, 0.1) {\footnotesize $\dot{\eta} = \Phi \eta + \Gamma e$};	
		\draw[-{Triangle}] (5, -0.2) -- (5, -2.5) -- (-2.1, -2.5) -- (-2.1, -0.5);
		\node[anchor = west] at (5, -1) {\footnotesize $\eta$};
        \node[anchor = west] at (-1.4, -1) {\footnotesize $\hat{\zeta}$};
        \filldraw[fill=gray!10, rounded corners=1mm](-2.5, -0.5) rectangle (-1, 0.5);
        \node[anchor = center] at (-1.75, 0) {\footnotesize $u = K\begin{bmatrix}
            \hat{\zeta} \\ \eta
        \end{bmatrix}$};
        \draw[-{Triangle}] (0, -1.75) -- (-1.4, -1.75) -- (-1.4, -0.5);
        \node[cloud,
		draw =black,
		text=black,
        text width = 1.1cm,
		fill = gray!10,
		aspect=2] (c) at (1.5, 2.15){\footnotesize Exosystem \\ $\;\,\dot{w} = Sw$};
        \draw[-{Triangle}] (1.5, 1.25) -- (1.5, 0.75);
		\node[anchor = west] at (1.5, 1.05) {\small{$w$}};
	\end{tikzpicture}
    
    \caption{Implementation of the controller \eqref{eq:controller_out_reg} designed in Algorithm \ref{alg:out_reg}.}
    \label{fig:out_reg_scheme}
    
\end{figure}

\begin{algorithm}[ht!]
    \caption{Data-Driven Output Regulation}\label{alg:out_reg}
    \begin{algorithmic}
    \State \hspace{-0.47cm} \textbf{Initialization} 
    \State \emph{Dataset:}
    \begin{equation}\label{eq:dataset_2}
        (u(t), e(t), y_r(t)), \qquad \forall t \in [0, \tau].
    \end{equation}
    \State \emph{Tuning:} $F$, $G$, $L$ such that \eqref{eq:non-minimal_plant} is a canonical non-minimal realization of \eqref{eq:plant}, with $L = [L_e\; L_{\text{r}}]$, $L_e \in \R^{\mu \times q}$, and $L_{\text{r}} \in \R^{\mu \times (p - q)}$; $(F_0, G_0)$ as in \eqref{eq:Delta}; $(S_0, \Gamma_0)$ as in \eqref{eq:Phi_Gamma}; number of samples $N \in \N$, $N \geq 1$, $\Ts \coloneqq \tau/N$.
    \State \hspace{-0.47cm} \textbf{Data Batches Construction} 
    \State \emph{Filter of the data:}
    \begin{equation}\label{eq:filter_out_reg}
        \dot{\hat{\zeta}}(t) = F\hat{\zeta}(t) + Gu(t) + L_e e(t) + L_{\text{r}} y_{\text{r}}(t),\quad \forall t \in [0, \tau].
    \end{equation}
    Initialization: $\hat{\zeta}(0) = 0 \in \R^{\mu}$.
    \State \emph{Internal model dynamics:}
    \begin{equation}\label{eq:IM_out_reg}
        \dot{\eta}(t) = \Phi \eta(t) + \Gamma e(t),\quad \forall t \in [0, \tau].
    \end{equation}
    Initialization: $\eta(0) = 0 \in \R^{dq}$.
    \State \emph{Auxiliary dynamics:}
    \begin{equation}\label{eq:aux_out_reg}
        \dot{\chi}(t) = \begin{bmatrix}
            S_0 & 0_{d \times \delta}\\
            0_{\delta \times d} & F_0
        \end{bmatrix}\chi(t),\quad \forall t \in [0, \tau].
    \end{equation}
    Initialization: $\chi(0) = \col(\Gamma_0, G_0) \in \R^{d + \delta}$.
    \State \emph{Sampled data batches:}
    \begin{equation}\label{eq:batches_out_reg}
        \begin{split}
            \mathcal{U}\! & \coloneqq  \!\begin{bmatrix}
                u(0) & u(\Ts) & \!\cdots\! & \,u((N - 1)\Ts)
            \end{bmatrix} \in \R^{m \times N}\\
            \mathcal{X}\! & \coloneqq\! \begin{bmatrix}
                \chi(0) & \!\chi(\Ts) & \!\cdots\! & \!\chi((N-1)\Ts)
            \end{bmatrix} \in \R^{(d + \delta) \times N}\\
            \mathcal{Z}\! & \coloneqq\! \begin{bmatrix}
                \hat{\zeta}(0) & \hat{\zeta}(\Ts) & \!\cdots\! & \hat{\zeta}((N-1)\Ts)\\
                \eta(0) & \eta(\Ts) & \!\cdots\! & \eta((N-1)\Ts)
            \end{bmatrix}\! \in \! \R^{(\mu + dq) \times N}\\
            \dot{\mathcal{Z}} \! & \coloneqq\! \begin{bmatrix}
                \dot{\hat{\zeta}}(0) & \dot{\hat{\zeta}}(\Ts) & \!\cdots\! & \dot{\hat{\zeta}}((N-1)\Ts)\\
                \dot{\eta}(0) & \dot{\eta}(\Ts) & \!\cdots\! & \dot{\eta}((N-1)\Ts)
            \end{bmatrix}\! \in \! \R^{(\mu + dq) \times N}.
        \end{split}
    \end{equation}
    \State \hspace{-0.47cm} \textbf{Stabilizing Gain Computation}
    \State \emph{LMI:} find $\mathcal{P} \in \R^{(\mu + dq) \times (\mu + dq)}$, $\mathcal{Q} \in \R^{N \times (\mu + dq)}$ such that:
    \begin{equation}\label{eq:LMI_out_reg}
        \begin{cases}
            \mathcal{P} = \mathcal{P}^\top \succ 0\\
            \dot{\mathcal{Z}}\mathcal{Q} + \mathcal{Q}^\top\dot{\mathcal{Z}}^\top \prec 0\\
            \begin{bmatrix}
                0_{\delta \times (\mu + dq)}\\
                \mathcal{P}
            \end{bmatrix} = \begin{bmatrix}
                \mathcal{X}\\
                \mathcal{Z}
            \end{bmatrix} \mathcal{Q}.
        \end{cases}
    \end{equation}
    \State \emph{Control gain:}
    \begin{equation}\label{eq:K_out_reg}
        K \! = \! \begin{bmatrix}
            K_\zeta\! & \!K_\eta
        \end{bmatrix}\! = \mathcal{U}\mathcal{Q}\mathcal{P}^{-1}, \quad K_\zeta \! \in \! \R^{m \times \mu}, K_\eta \! \in \! \R^{m \times dq}.
    \end{equation}
    \State \hspace{-0.47cm} \textbf{Control Deployment}
    \State \emph{Control law:}
    \begin{equation}\label{eq:controller_out_reg}
        \dot{\xi} = \begin{bmatrix}
            F \! + \!GK_\zeta & GK_\eta \\
            0_{dq \times \mu} & \Phi
        \end{bmatrix} \xi + \begin{bmatrix}
            L_e & L_{\text{r}}\\
            \Gamma & 0_{dq \times (p - q)}\!
        \end{bmatrix}\!\begin{bmatrix}
            e \\ y_{\text{r}}
        \end{bmatrix}\!,\;  u = K\xi.
    \end{equation}
    Initialization: $\xi(0) \in \R^{\mu + dq}$ arbitrary.
    \end{algorithmic}    
\end{algorithm}

Following the classical literature of output regulation \cite{davison1976robust, isidori2003robust, isidori2017lectures}, the problem of Section \ref{sec:problem2} is solved by a controller that embeds an internal model unit.
In particular, let
\begin{equation}
    m_s(s) \coloneqq s^d + \theta_{s,d-1} s^{d-1} + \ldots + \theta_{s, 1}s + \theta_{s, 0}
\end{equation}
be the minimal polynomial of $S$.
Then, define:
\begin{equation}\label{eq:Phi_Gamma}
    \Phi \! \coloneqq \! I_q \otimes \underbrace{\begin{bmatrix}
        0 & 1 & \cdots & 0\\
        \vdots & & \ddots & \vdots\\
        0 & 0 & \cdots & 1\\
        -\theta_{s, 0} &\! -\theta_{s,1} & \cdots &\! -\theta_{s, d-1}
    \end{bmatrix}}_{\eqqcolon S_0}, \;\; \Gamma \! \coloneqq I_q \otimes \underbrace{\begin{bmatrix}
        0\\
        \vdots\\
        0\\
        \omega_s
    \end{bmatrix}}_{\eqqcolon \Gamma_0}\!,
\end{equation}
where $\omega_s \neq 0$ is a scalar tuning gain.
A \emph{post-processing internal model} is the following system:
\begin{equation}\label{eq:IM}
    \dot{\eta} = \Phi \eta + \Gamma e,
\end{equation}
where $\eta \in \R^{dq}$.
It is known that, under Assumptions \ref{hyp:ctrb_obs} and \ref{hyp:non-resonance}, the augmented plant
\begin{equation}\label{eq:aug_plant}
    \begin{split}
        \begin{bmatrix}
            \dot{x}\\
            \dot{\eta}
        \end{bmatrix} &= \begin{bmatrix}
            A & 0_{n \times dq}\\
            \Gamma C_e & \Phi
        \end{bmatrix}\begin{bmatrix}
            x \\ \eta
        \end{bmatrix} + \begin{bmatrix}
            B \\ 0_{dq \times m}
        \end{bmatrix} u + \begin{bmatrix}
            P \\ \Gamma Q_e
        \end{bmatrix}w\\
        y_{\text{a}} &= \begin{bmatrix}
            C & 0_{p \times dq}\\
            0_{dq \times n} & I_{dq}
        \end{bmatrix}\begin{bmatrix}
            x \\ \eta
        \end{bmatrix} + \begin{bmatrix}
            Q \\ 0_{dq \times l}
        \end{bmatrix}w,
    \end{split}
\end{equation}
can be stabilized via output feedback when $w = 0$ \cite[Lem. 4.2]{isidori2017lectures}.
Let
\begin{equation}\label{eq:stabilizer}
    \begin{split}
        \dot{\xi}_{\text{s}} &= A_{\text{s}} \xi_{\text{s}} + B_{\text{s}} y_{\text{a}}\\
        u &= C_{\text{s}}\xi_{\text{s}} + D_{\text{s}} y_{\text{a}}
    \end{split}
\end{equation}
be a controller that makes the origin $(x, \eta, \xi_{\text{s}}) = 0$ of the interconnection \eqref{eq:aug_plant}, \eqref{eq:stabilizer} asymptotically stable when $w = 0$.
Then, it is known that the output regulation problem as stated in Section \ref{sec:problem2} is solved by a controller that combines \eqref{eq:IM} and \eqref{eq:stabilizer} \cite[Prop. 4.3]{isidori2017lectures}.
Therefore, the remainder of this section involves extending the stabilization approach of Section \ref{sec:stabilization} to the construction of \eqref{eq:stabilizer}.

\begin{remark}
    We give an intuition on how the approach is solved via \cite[Prop. 4.3]{isidori2017lectures}.
    Given the properties of the stabilizer \eqref{eq:stabilizer}, the closed-loop interconnection of the plant \eqref{eq:plant_w}, the exosystem \eqref{eq:exo}, and the controller \eqref{eq:IM}, \eqref{eq:stabilizer} has a stable and a center eigenspace.
    In particular, for some matrices $\Sigma_x$, $\Sigma_\eta$, $\Sigma_\xi$, the closed-loop solutions satisfy $x(t) \to \Sigma_x w(t)$, $\eta(t) \to \Sigma_\eta w(t)$, $\xi_{\text{s}}(t) \to \Sigma_\xi w(t)$.
    Then, given the special choice \eqref{eq:Phi_Gamma} for the matrices $(\Phi, \Gamma)$ used in the internal model \eqref{eq:IM}, the following property holds:
    \begin{equation}
        C_e\Sigma_x + Q_e = 0,
    \end{equation}
    which implies $e(t) \to C_e\Sigma_xw(t) + Q_e w(t) = 0$.
    We refer to \cite[Ch. 4]{isidori2017lectures}, for an in-depth presentation of these topics.
\end{remark}
Under Assumption \ref{hyp:ctrb_obs}, let $F$, $G$, and $L$ be such that \eqref{eq:non-minimal_plant} is a canonical non-minimal realization of \eqref{eq:plant}, for some $\Pi$ and $H$ such that the matrix equation \eqref{eq:Pi} holds.
Also, let $H_e \coloneqq C_e \Pi$.
Then, as before, introduce system \eqref{eq:obs} and the error $\epsilon$ in \eqref{eq:epsilon}.
Note that the system to be controlled and from which we collect the data is \eqref{eq:plant_w}, which is perturbed by the exogenous signal $w$.
In particular, using \eqref{eq:epsilon}, we can interconnect the plant \eqref{eq:plant_w}, the exosystem \eqref{eq:exo}, the filter \eqref{eq:obs}, and the internal model \eqref{eq:IM} to obtain the following dynamics:
\begin{equation}\label{eq:w_eps_zeta_eta}
    \begin{bmatrix}
        \dot{w} \\ \dot{\epsilon} \\ \dot{\hat{\zeta}} \\ \dot{\eta}
    \end{bmatrix} \! = \! \left[\begin{array}{c|c}
        \!\!\!\!\!\begin{array}{cc}
            S \! & \! 0_{l \times n} \\
            P \! - \! \Pi LQ \! & \! A \! - \! \Pi L C
        \end{array} \!\!\!\! & \!\!\!\! 0_{(l + n)\times(\mu + dq)} \!\!\!\!\!\\
        \hline
        \!\!\!\!\!\begin{array}{cc}
            0_{\mu \times l} \;\; & \;\; LC \\
            \Gamma Q_e \;\; & \;\; \Gamma C_e 
        \end{array} \!\!\!\! & \!\!\!\! \begin{array}{cc}
            F \! + \! LH \! & \! 0_{\mu \times dq} \\
            \Gamma H_e \! & \! \Phi
        \end{array}\!\!\!\!\!
    \end{array} \right]\!\begin{bmatrix}
        w \\ \epsilon \\ \hat{\zeta} \\ \eta
    \end{bmatrix} + \begin{bmatrix}
        0_{l \times m} \\ 0_{n \times m} \\ G \\ 0_{dq \times m}
    \end{bmatrix}\! u,
\end{equation}
which have the same structure of \eqref{eq:eps_zeta}, with $\epsilon$ replaced by $\col(w, \epsilon)$ and $\hat{\zeta}$ replaced by $\col(\hat{\zeta}, \eta)$.

We now study some structural properties of \eqref{eq:w_eps_zeta_eta}.
The first result is the extension of the non-resonance condition of Assumption \ref{hyp:non-resonance} from the original plant \eqref{eq:plant} to the non-minimal realization \eqref{eq:non-minimal_plant}.
\begin{lemma}\label{lem:non-resonance}
    Under Assumptions \ref{hyp:ctrb_obs} and \ref{hyp:non-resonance}, it holds that
    \begin{equation}\label{eq:non-min_non-resonance}
        \rank \begin{bmatrix}
            F + LH - sI_{\mu} & G\\
            H_e & 0_{q \times m}
        \end{bmatrix} = \mu + q,
    \end{equation}
    for all $s \in \sigma(S)$.
\end{lemma}
\begin{proof}
    By Lemma \ref{lem:Pi}, there exists a similarity transformation $T$ such that system \eqref{eq:non-minimal_plant} is decomposed as in \eqref{eq:Kalman_decomposition}.
    Then,
    \begin{equation}
        \begin{split}
            \bdiag(T, I_q)\begin{bmatrix}
                F + LH - sI_{\mu} & G\\
                H_e & 0_{q \times m}
            \end{bmatrix} \bdiag(T^{-1}, I_m)\\ = \begin{bmatrix}
                A_{\bar{o}} - sI_{\mu - n} & A_{\times} & B_{\bar{o}}\\
                0_{n \times (\mu - n)} & A - sI_n & B\\
                0_{q \times (\mu - n)} & C_e & 0_{q \times m}
            \end{bmatrix},
        \end{split}
    \end{equation}
    where $A_{\bar{o}}$ is Hurwitz by Lemma \ref{lem:st_det}.
    The proof follows by recalling Assumption \ref{hyp:non-resonance} and noting that $\rank (A_{\bar{o}} - sI_{\mu - n}) = \mu - n$ for all $s \in \sigma(S)$.
\end{proof}

An important consequence of Lemma \ref{lem:non-resonance} is the following corollary, whose proof is omitted for brevity as it follows the same lines of the proof of \cite[Lemma 4.2]{isidori2017lectures}.
\begin{corollary}\label{cor:stab_ctrb}
    Under Assumptions \ref{hyp:ctrb_obs} and \ref{hyp:non-resonance}, let \eqref{eq:non-minimal_plant} be a weak (resp. strong) canonical non-minimal realization of \eqref{eq:plant}.
    Then, the pair
    \begin{equation}\label{eq:aug_pair}
        \left(\begin{bmatrix}
            F + LH & 0_{\mu \times dq}\\
            \Gamma H_e & \Phi
        \end{bmatrix}, \begin{bmatrix}
            G \\ 0_{dq \times m}
        \end{bmatrix} \right)
    \end{equation}
    is stabilizable (resp. controllable).
\end{corollary}
Corollary \ref{cor:stab_ctrb} ensures that the stabilizer \eqref{eq:stabilizer} can be obtained by combining \eqref{eq:obs} and a feedback law of the form $u = K_\zeta\hat{\zeta} + K_\eta \eta$.

Therefore, the only remaining step is to design $K = [K_{\zeta}\;\; K_{\eta}]$ using the available dataset.
First, we simulate the filter and the internal model unit as in \eqref{eq:filter_out_reg} and \eqref{eq:IM_out_reg}.
Then, to be able to use \eqref{eq:w_eps_zeta_eta} for data-driven control, we need to overcome the limitation that the signals $\col(w(t), \epsilon(t))$ are not available for measurement.
To this aim, note that the matrix
\begin{equation}\label{eq:aug_exo}
    \begin{bmatrix}
        S & 0_{l \times n}\\
        P - \Pi L Q & A - \Pi L C
    \end{bmatrix}
\end{equation}
is such that $\sigma(S) \cap \sigma(A - \Pi L C) = \emptyset$ due to Lemma \ref{lem:min_poly}, which implies that the Jordan normal form of \eqref{eq:aug_exo} is block diagonal, with one block associated with $S$ and the other with $A - \Pi L C$.
Thus, we can extend Lemma \ref{lem:chi} to ensure that $\col(w(t), \epsilon(t)) = D\chi(t)$, where $\chi(t)$ is obtained by simulating \eqref{eq:aux_out_reg}.
The statement of such result and its proof are omitted for brevity as they are identical to Lemma \ref{lem:chi}.

Next, as before, we sample the available continuous-time signals in $N$ distinct time instants (with fixed sampling time $\Ts$ for simplicity) to obtain the data batches \eqref{eq:batches_out_reg}.
Finally, we employ the LMI \eqref{eq:LMI_out_reg} and equation \eqref{eq:K_out_reg} to compute the stabilizing gain $K = [K_\zeta \; K_\eta]$.
The overall controller combining the internal model and the stabilizer is given in \eqref{eq:controller_out_reg}, which we report in the same form as \eqref{eq:ctrl_w}.

The following result provides the guarantees for Algorithm \ref{alg:out_reg}.
Its proof is omitted as it is based on Theorem \ref{thm:stabilization} and the arguments above. See \cite[\S 4.4]{isidori2017lectures} for further details.
\begin{theorem}
    Consider Algorithm \ref{alg:out_reg} and let Assumptions \ref{hyp:ctrb_obs} and \ref{hyp:non-resonance} hold.
    Then:
    \begin{enumerate}
        \item The LMI \eqref{eq:LMI_out_reg} is feasible if
        \begin{equation}\label{eq:E_out_reg}
            \rank\begin{bmatrix}
                \mathcal{X} \\ \mathcal{Z} \\ \mathcal{U}
            \end{bmatrix} = 
            (d + \delta) + (\mu + dq) + m.
        \end{equation}
        \item For any solution $\mathcal{P}$, $\mathcal{Q}$ of \eqref{eq:LMI_out_reg}, the gain $K$ computed from \eqref{eq:K_out_reg} is such that
        \begin{equation}
            \begin{bmatrix}
            F + LH & 0_{\mu \times dq}\\
            \Gamma H_e & \Phi
        \end{bmatrix} + \begin{bmatrix}
            G \\ 0_{dq \times m}
        \end{bmatrix} K
        \end{equation}
        is Hurwitz.
        Thus, the controller \eqref{eq:controller_out_reg} solves the data-driven output regulation problem of Section \ref{sec:problem2}.
    \end{enumerate}
\end{theorem}

\section{Design of Canonical Non-Minimal Realizations}\label{sec:tuning}

In Sections \ref{sec:stabilization} and \ref{sec:out_reg}, we have shown that data-driven output-feedback control is feasible as long as $F$, $G$, $L$ are known such that \eqref{eq:non-minimal_plant} is a canonical non-minimal realization of the plant \eqref{eq:plant}, for some output matrix $H$.
It is therefore convenient to provide tuning techniques for such matrices.
In this section, we show that the design of $F$, $G$, and $L$ can be derived from some structural properties of the pair $(C, A)$.
Furthermore, we illustrate how such properties can be inferred from input-output data of the form \eqref{eq:dataset}, as long as a full-rank condition similar to \eqref{eq:E} holds.
We begin by illustrating the tuning in the two simplified scenarios of \cite{bosso2024derivative}.

\subsection{State-Feedback Scenario (C = I\texorpdfstring{\textsubscript{n}}{n})}
Let Assumption \ref{hyp:ctrb_obs} hold.
Also, suppose that the order $n$ of system \eqref{eq:plant} is known and that we have full access to its states, i.e., $y = x$.
Then, we can assign the matrices of \eqref{eq:non-minimal_plant} as follows:
\begin{equation}\label{eq:FGL_state}
    F = \begin{bmatrix}
        -\lambda I_n & 0_{n \times m}\\
        0_{m \times n} & -\lambda I_m
    \end{bmatrix}, \quad G = \begin{bmatrix}
        0_{n \times m} \\
        \gamma I_m
    \end{bmatrix}, \quad L = \begin{bmatrix}
        \gamma I_n\\
        0_{m \times n}
    \end{bmatrix},
\end{equation}
where $\lambda > 0$ and $\gamma \neq 0$ are scalar tuning gains.

Using $C = I_n$, we obtain that \eqref{eq:Pi} is solved by
\begin{equation}\label{eq:H_state}
    \Pi = H = \gamma^{-1}\begin{bmatrix}
        A + \lambda I_n & B
    \end{bmatrix},
\end{equation}
thus we conclude that \eqref{eq:non-minimal_plant} with $F$, $G$, $L$ in \eqref{eq:FGL_state} and $H$ in \eqref{eq:H_state} is a canonical non-minimal realization.
Also, note that
\begin{equation}
    F + LH = \begin{bmatrix}
        A & B\\
        0_{m \times n} & -\lambda I_m
    \end{bmatrix},
\end{equation}
and we can prove that $(F + LH, G)$ is controllable (i.e., the realization is strong) because, using the PBH test, the condition
\begin{equation}
    \rank\begin{bmatrix}
        sI_n - A & -B & 0_{n \times m}\\
        0_{m \times n} & sI_m + \lambda I_m & \gamma I_m
    \end{bmatrix} = n + m
\end{equation}
is verified for all $s \in \sigma(A) \cup \{-\lambda\}$.

\subsection{SISO Output-Feedback Scenario (m = p = 1)}
Suppose again that the order $n$ of system \eqref{eq:plant} is known.
Then, we follow the structure adopted in classical adaptive observer design (see \cite{kreisselmeier1977adaptive} or ``representation 2'' of \cite[Ch. 4]{narendra1989stable}) by letting
\begin{equation}\label{eq:FGL_siso}
    F = \begin{bmatrix}
        \Lambda & 0_{n \times n}\\
        0_{n \times n} & \Lambda
    \end{bmatrix}, \quad G = \begin{bmatrix}
        0_{n \times 1} \\
        \ell
    \end{bmatrix}, \quad L = \begin{bmatrix}
        \ell \\
        0_{n \times 1}
    \end{bmatrix},
\end{equation}
where $(\Lambda, \ell)$ is a controllable pair, and matrix $\Lambda \in \R^{n \times n}$ is Hurwitz and has $n$ distinct eigenvalues.
As stated in the following results, the tuning \eqref{eq:FGL_siso} implies the existence of a strong canonical non-minimal realization.
We postpone the proofs to Theorems \ref{thm:tuning_nu} and \ref{thm:ctrb_nu}, which generalize these statements.
\begin{proposition}
    Under Assumption \ref{hyp:ctrb_obs}, the matrices $F$, $G$, and $L$ in \eqref{eq:FGL_siso} are such that there exist matrices $\Pi$ and $H$ solving equation \eqref{eq:Pi}.
    As a consequence, for some $H$, \eqref{eq:non-minimal_plant} is a canonical non-minimal realization of system \eqref{eq:plant}.
\end{proposition}
\begin{proposition}
    Let Assumption \ref{hyp:ctrb_obs} hold.
    Consider $F$, $G$, $L$ as in \eqref{eq:FGL_siso}, and let $\Pi$ and $H$ be solutions to the matrix equation \eqref{eq:Pi}.
    Then, the pair $(F + LH, G)$ is controllable.
\end{proposition}

\subsection{MIMO Output-Feedback Scenario}\label{sec:mimo}
In the previous examples, we have shown that a strong non-minimal realization can be achieved if the order of the plant is known.
It turns out that the structures in \eqref{eq:FGL_state} and \eqref{eq:FGL_siso} are more precisely related to the observability indices of system \eqref{eq:plant}, for which we briefly recall the basic notions.
For further details, we refer to \cite[Ch. 3]{antsaklis1997linear}.
In \eqref{eq:plant}, let
\begin{equation}
    C = \begin{bmatrix}
        c_1 &
        \cdots &
        c_p
    \end{bmatrix}^{\top}\!\!,
\end{equation}
and suppose that $\rank C = p$, so that $c_1, \ldots, c_p \in \R^n$ are linearly independent.
Then, consider the following sequence:
\begin{equation}
    c_1, \ldots, c_p, A^\top\! c_1, \ldots, A^\top\! c_p, \ldots, (A^{\top})^n c_1, \ldots, (A^{\top})^n c_p
\end{equation}
and select, starting from the left and moving to the right, the first appearing $n$ linearly independent vectors (which exist due to Assumption \ref{hyp:ctrb_obs}).
Such vectors can be reordered as follows:
\begin{equation}
    c_1, A^\top c_1, \ldots, (A^{\top})^{\nu_1 - 1}c_1, \ldots, c_p, A^\top c_p, \ldots, (A^{\top})^{\nu_p - 1}c_p.
\end{equation}
The integers $\nu_1, \ldots, \nu_p$ are called the \emph{observability indices} of system \eqref{eq:plant}, while $\nu \coloneqq \max_i \nu_i$ is the \emph{observability index} of \eqref{eq:plant}.
Note that $\nu_i \geq 1$ due to $\rank C = p$ and, by construction,
\begin{equation}
    \sum_{i = 1}^{p}\nu_i = n \leq \nu p.
\end{equation}
In the remainder of Section \ref{sec:mimo} and in Section \ref{sec:index}, we obtain results under the following assumption, which imposes a uniform observability index across all outputs.
\begin{assumption}\label{hyp:nu}
    It holds that $\nu_1 = \nu_2 = \ldots = \nu_p = \nu$.
\end{assumption}
Assumption \ref{hyp:nu} covers the previous state-feedback ($\nu_i = 1$, for all $i \in \{1, \ldots, n\}$) and SISO output-feedback ($\nu_1 = \nu = n$) cases.
Another relevant example is given by multi-input single-output (MISO) systems.
We also provide MIMO examples in Section \ref{sec:simulations}.
Future work will be dedicated to studying the general case of arbitrary observability indices.

Under Assumption \ref{hyp:nu}, we obtain the following two theorems, which generalize and unify the state-feedback and SISO output-feedback scenarios.
Since the proofs are quite lengthy, they are deferred to Appendix \ref{sec:proofs}.
\begin{theorem}\label{thm:tuning_nu}
    Let Assumptions \ref{hyp:ctrb_obs} and \ref{hyp:nu} hold.
    Pick $\Lambda \in \R^{\nu \times \nu}$ and $\ell \in \R^{\nu}$ such that $(\Lambda, \ell)$ is controllable and $\Lambda$ is Hurwitz and has $\nu$ distinct eigenvalues.
    Choose
    \begin{equation}\label{eq:FGL_nu}
        F = I_{p+m} \otimes \Lambda, \qquad G = \begin{bmatrix}
            0_{p\nu \times m}\\
            I_m \otimes \ell
        \end{bmatrix},\qquad L = \begin{bmatrix}
            I_p \otimes \ell\\
            0_{m\nu \times p}
        \end{bmatrix}.
    \end{equation}
    Then, there exist matrices $\Pi$ and $H$ such that:
    \begin{itemize}
        \item Equation \eqref{eq:Pi} holds.
        \item $A - \Pi LC$ is similar to $I_p \otimes \Lambda$.
    \end{itemize}
    As a consequence $F$, $G$, and $L$ are such that \eqref{eq:non-minimal_plant}, for some $H$, is a canonical non-minimal realization of \eqref{eq:plant}.
\end{theorem}
\begin{remark}
    The proof of Theorem \ref{thm:tuning_nu} involves two steps.
    First, using the multivariable observer canonical form (see Appendix \ref{sec:MIMO_obs_form}), we show there exists a linear equation whose solutions are also solutions to the quadratic equation \eqref{eq:Pi}.
    Then, we exploit the tuning \eqref{eq:FGL_nu} to decouple the linear equation into elementary blocks that can be solved via Lemma \ref{lem:X}, given in Appendix \ref{sec:X}.
\end{remark}
\begin{theorem}\label{thm:ctrb_nu}
    Given the assumptions and matrices $F$, $G$ and $L$ of Theorem \ref{thm:tuning_nu}, let $\Pi$ and $H$ be solutions to the matrix equation \eqref{eq:Pi}.
    Then, the pair $(F + LH, G)$ is controllable.
\end{theorem}
\begin{remark}
    Also the proof of this result is divided into two steps.
    First, we relate the input-output models of \eqref{eq:plant} and \eqref{eq:non-minimal_plant} using Lemma \ref{lem:Pi} to obtain a rank condition for the polynomial matrices of \eqref{eq:non-minimal_plant}.
    Then, we show that this rank condition implies controllability of $(F + LH, G)$ using the PBH test.
\end{remark}
In view of the tuning \eqref{eq:FGL_nu}, the systems employed in \eqref{eq:filter} and \eqref{eq:filter_out_reg} have the following structure:
\begin{equation}
    \dot{\hat{\zeta}} = \begin{bmatrix}
        \Lambda & & \\
        &\ddots & \\
        & & \Lambda
    \end{bmatrix}\hat{\zeta} + \begin{bmatrix}
        \ell & & \\
        &\ddots & \\
        & & \ell
    \end{bmatrix}\begin{bmatrix}
        y \\ u
    \end{bmatrix},
\end{equation}
which consists of $p + m$ parallel asymptotically stable filters of dimension $\nu$, one for each entry of $\col(y, u)$.
It is notable that this form is conceptually the same of the internal model \eqref{eq:IM}, which can be written more explicitly as:
\begin{equation}
    \dot{\eta} = \begin{bmatrix}
        S_0 & & \\
        &\ddots & \\
        & & S_0
    \end{bmatrix}\eta + \begin{bmatrix}
        \Gamma_0 & & \\
        &\ddots & \\
        & & \Gamma_0
    \end{bmatrix}e,
\end{equation}
and consists of $q$ parallel neutrally stable filters of dimension $d$, one for each entry of the regulated output $e$.

We also note that the minimal polynomial of $F$ coincides with the characteristic polynomial of $\Lambda$, hence we can replace \eqref{eq:aux} in Algorithm \ref{alg:stabilization} with
\begin{equation}\label{eq:aux_Lambda}
    \dot{\chi}(t) = \Lambda \chi(t), \qquad \chi(0) = \ell.
\end{equation}
Similarly, \eqref{eq:aux_out_reg} in Algorithm \ref{alg:out_reg} can be replaced with
\begin{equation}
    \dot{\chi}(t) = \begin{bmatrix}
        S_0 & 0_{d \times \nu}\\
        0_{\nu \times d} & \Lambda
    \end{bmatrix}\chi(t), \qquad \chi(0) = \begin{bmatrix}
        \Gamma_0\\
        \ell
    \end{bmatrix}.
\end{equation}

\subsection{Observability Index Estimation}\label{sec:index}

Under Assumption \ref{hyp:nu}, we have seen that a strong canonical non-minimal realization can be constructed if $\nu$ is known.
However, $\nu$ is usually not available from prior information.
Therefore, we now propose an approach to infer it from the data.
The procedure is summarized in Algorithm \ref{alg:index} and described below for the data-driven stabilization problem of Section \ref{sec:problem1}.
A similar procedure can be derived for data-driven output regulation but is omitted for brevity.

\begin{algorithm}[b!]
    \caption{Observability Index Estimation (Stabilization)}\label{alg:index}
    \begin{algorithmic}
    \State \hspace{-0.47cm} \textbf{Initialization} 
    \State \emph{Dataset:}
    \begin{equation}\label{eq:dataset_nu_hat}
        (u(t), y(t)), \qquad \forall t \in [0, \tau].
    \end{equation}
    \State \emph{Tuning:} scalars $0 < \lambda_1 < \ldots < \lambda_{\nu_{\max}}$ and $\gamma_1 \neq 0, \ldots, \gamma_{\nu_{\max}} \neq 0$, with $\nu_{\max} \in \N$, $\nu_{\max} > \nu$; number of samples $N \in \N$, $N \geq 1$, $\Ts \coloneqq \tau/N$; initial guess $\hat{\nu} = 1$.
    \State \hspace{-0.47cm} \textbf{Index Estimation}
    \While{\emph{index not found} {\bf and} $\hat{\nu} < \nu_{\max}$}
        \State $\hat{\nu} \gets \hat{\nu} + 1$
        \State \emph{Matrix gains:}
            \begin{equation}\label{eq:Lambda_ell_hat}
                \quad\;\;\; \hat{\Lambda} \coloneqq -\diag(\lambda_1, \ldots, \lambda_{\hat{\nu}}),\quad \hat{\ell} \coloneqq \col(\gamma_1, \ldots, \gamma_{\hat{\nu}}).
            \end{equation}
        \State \emph{Simulation for $t \in [0, \tau]$:}
            \begin{equation}\label{eq:filter_nu_hat}
                \begin{split}
                    \dot{\chi}(t) &= \hat{\Lambda}\chi(t)\\
                    \dot{\hat{\zeta}}(t) &= (I_{p+m} \otimes \hat{\Lambda})\hat{\zeta}(t) + (I_{p + m} \otimes \hat{\ell})\begin{bmatrix}
                    y(t) \\ u(t)
                \end{bmatrix}.
                \end{split}
            \end{equation}
        \State Initialization: $\hat{\chi}(0) = \hat{\ell}$, $\hat{\zeta}(0) = 0_{\hat{\nu}(p + m) \times 1}$.
        \State \emph{Sampled data batch:}
        \begin{equation}\label{eq:batches_nu_hat}
            \mathcal{B} \coloneqq \begin{bmatrix}
            \chi(0) & \chi(\Ts) & \!\cdots\! & \chi((N-1)\Ts)\\
            \hat{\zeta}(0) & \hat{\zeta}(\Ts) & \!\cdots\! & \hat{\zeta}((N-1)\Ts)\\
            \end{bmatrix}.
        \end{equation}
        \If{$\rank \mathcal{B} < \hat{\nu}(p + m + 1)$}
            \State \emph{index found} \Comment{Assumes $\mathcal{B}$ full rank if $\hat{\nu}= \nu$. \hspace{3pt} }
        \EndIf
    \EndWhile
    \State $\hat{\nu} \gets \hat{\nu} - 1$ \Comment{Observability index estimate. \hspace{5pt} }
    \end{algorithmic}
\end{algorithm}

Suppose that a dataset of the form \eqref{eq:dataset} (reported also in \eqref{eq:dataset_nu_hat} for convenience) is available.
Then, given some estimate $\hat{\nu}$ of $\nu$, consider the matrices $\hat{\Lambda} \in \R^{\hat{\nu}\times\hat{\nu}}$ and $\hat{\ell} \in \R^{\hat{\nu}}$ as in \eqref{eq:Lambda_ell_hat}, and use them to replace $\Lambda$ and $\ell$ of \eqref{eq:FGL_nu}.
By Theorem \ref{thm:tuning_nu}, $\hat{\nu} = \nu$ implies the existence of a non-minimal realization.
It turns out that the same holds for any $\hat{\nu} \geq \nu$ in view of the next lemma, whose proof is obtained via direct verification.
\begin{lemma}\label{lem:redundancy}
    Let $\Pi$ and $H$ solve equation \eqref{eq:Pi}, for some $A$, $B$, $C$, and $F$, $G$, $L$.
    Suppose that $F$, $G$, and $L$ are replaced by
    \begin{equation}
        \hat{F} \coloneqq T\bdiag(F, \tilde{F})T^{-1}, \quad \hat{G} \coloneqq T\begin{bmatrix}
            G \\ \tilde{G}
        \end{bmatrix}, \quad \hat{L} \coloneqq T\begin{bmatrix}
            L \\ \tilde{L}
        \end{bmatrix},
    \end{equation}
    with $\tilde{F}$, $\tilde{G}$, $\tilde{L}$ arbitrary matrices and $T$ an invertible matrix.
    Then, equation \eqref{eq:Pi} is solved by $\hat{\Pi} \coloneqq [\Pi\;\; 0_{n \times \tilde{\mu}}]T^{-1}$ and $\hat{H} \coloneqq C\hat{\Pi} = [H\;\; 0_{p \times \tilde{\mu}}]T^{-1}$, where $\tilde{\mu}$ is the dimension of $\tilde{F}$.
\end{lemma}
For a given $\hat{\nu}$, Algorithm \ref{alg:index} involves simulating \eqref{eq:filter_nu_hat} and collecting a batch of sampled data of the form \eqref{eq:batches_nu_hat}.

Suppose that, for $\hat{\nu} = \nu$, $\mathcal{B}$ has full rank (which is necessary for \eqref{eq:E}).
Then, the same is true for $\hat{\nu} \in \{1, \ldots, \nu\}$.
On the other hand, the rank is lost for $\hat{\nu} \geq \nu + 1$, as stated in the next result.
The proof is quite long and thus given in Appendix \ref{sec:proofs}.
\begin{theorem}\label{thm:nu_hat}
    Consider Algorithm \ref{alg:index} and let Assumptions \ref{hyp:ctrb_obs} and \ref{hyp:nu} hold.
    Then, for all $\hat{\nu} \geq \nu + 1$, the data in \eqref{eq:batches_nu_hat} satisfy:
    \begin{equation}
        \rank \mathcal{B} < \hat{\nu}(p + m + 1).
    \end{equation}
\end{theorem}
We conclude that the observability index can be computed by increasing $\hat{\nu}$ from $1$ until the rank of $\mathcal{B}$ is lost.
Note that the upper bound $\nu_{\max}$ is only used as stopping criterion in case the search fails.
Also note that, in each iteration, it is sufficient to simulate only the new filters $\dot{\chi}_{\hat{\nu}}(t) = -\lambda_{\hat{\nu}}\chi_{\hat{\nu}}(t)$ and $\dot{\hat{\zeta}}_{\hat{\nu}}(t) = -\lambda_{\hat{\nu}}\hat{\zeta}_{\hat{\nu}}(t) + \gamma_{\hat{\nu}}\col(y(t), u(t))$.
\begin{remark}
    A similar algorithm can be obtained in the data-driven output regulation setting.
    In particular, in \eqref{eq:filter_nu_hat}, the dynamics of $\chi$ are augmented as in \eqref{eq:aug_exo} and the internal model \eqref{eq:IM_out_reg} is included.
    A result corresponding to Theorem \ref{thm:nu_hat} can then be obtained by following, mutatis mutandis, the same proof.
\end{remark}

\section{Numerical Examples}\label{sec:simulations}
To illustrate the approaches of the previous sections, we present some numerical examples developed in MATLAB.
In particular, Algorithms \ref{alg:stabilization} and \ref{alg:out_reg} have been implemented with YALMIP \cite{Lofberg2004} and MOSEK \cite{mosek} to solve the LMIs.
The code is available at the linked repository\footnote{\url{https://github.com/IMPACT4Mech/continuous-time_data-driven_control}}.

\subsection{Unstable Batch Reactor Control}
We address the data-driven stabilization problem for the continuous-time linearized model of a batch reactor given in \cite{walsh2001scheduling}.
The matrices corresponding to system \eqref{eq:plant} are
\begin{equation}
    \begin{split}
        A &= \begin{bmatrix*}[r]
            1.38&\; -0.2077&\; 6.715&\; -5.676\\
            -0.5814&\; -4.29&\; 0&\; 0.675\\
            1.067&\; 4.273&\; -6.654&\; 5.893\\
            0.048&\; 4.273&\; 1.343&\; -2.104
        \end{bmatrix*}\\
        B &= \begin{bmatrix*}[r]
            0& 0\\
            5.679& 0\\
            1.136& -3.146\\
            1.136& 0
        \end{bmatrix*}\!,\quad\;\, C=\begin{bmatrix*}[r]
            1 & 0 & 1 & -1\\
            0 & 1 & 0 & 0
        \end{bmatrix*}\!.
    \end{split}
\end{equation}
It can be verified that $A$, $B$, and $C$ satisfy Assumptions \ref{hyp:ctrb_obs} and \ref{hyp:nu}, with observability index $\nu = 2$.

First, we acquire a dataset of the form $(u(t), y(t))$ by simulating the system for $\tau = 2$ s.
The applied input $u$ is the sum of $4$ sinusoids in both entries, with $8$ distinct frequencies.
The initial condition $x(0)$ is chosen randomly, with each entry extracted from the uniform distribution $U(-1, 1)$.
In the following, we describe a procedure that has been extensively tested with different initial conditions $x(0)$.

Given a dataset of the form \eqref{eq:dataset}, we run Algorithm \ref{alg:index} with $N = 50$ (so that $\Ts = 40$ ms) and, for each $j \in \N$ with $j \geq 1$, $\lambda_j = \gamma_j = j$.
For each tested dataset, we obtain $\hat{\nu} = \nu = 2$.

Next, we design a controller of the form \eqref{eq:controller} by running Algorithm \ref{alg:stabilization} with $N = 50$ and $F$, $G$, and $L$ as in \eqref{eq:FGL_nu} with $\Lambda = \diag(-4, -8)$ and $\ell = \col(1, 2)$.
Also, we use \eqref{eq:aux_Lambda} in place of \eqref{eq:aux}.
For each randomly generated dataset, the procedure succeeds in stabilizing the system.
For instance, with $x(0) = [-0.149\;\; 0.2225\;\; 0.7115\;\; 0.3416]^\top$, we obtain the gain
\begin{equation}
    K = \begin{bsmallmatrix*}[r]
        25.317&\;   -5.217&\;   12.549&\;   19.378&\;  -90.498&\;   49.304&\;    6.599&\;  -18.314\\
        14.095&\;    0.289&\;   15.612&\;   16.678&\;  -81.084&\;   56.432&\;   -3.329&\;    3.954
    \end{bsmallmatrix*},
\end{equation}
which places the eigenvalues of the closed-loop interconnection of \eqref{eq:plant} and \eqref{eq:controller} in $\{-0.901, -1.546 \pm 2.833\mathrm{i}, -2.106 \pm 32.492\mathrm{i}, -2.164, -4, -4, -4.261, -8.349, -8, -8\}$.

Using the same dataset, we also develop a controller with integral action.
In particular, this design corresponds to solving the output regulation problem with $S = 0$, $p = q = 2$, and the previously acquired output trajectory $y(t)$ treated as regulated output $e(t)$ (i.e., $w(t) = 0$).
It can be verified that, for the given plant and exosystem matrices, Assumption \ref{hyp:non-resonance} holds.
Therefore, we run Algorithm \ref{alg:out_reg} with $N$, $\Lambda$, and $\ell$ as in the previous case, while for the internal model we let $S_0 = 0$, $\Gamma_0 = 5$, obtaining $\Phi = 0_{2 \times 2}$ and $\Gamma = 5I_2$.
For each dataset, the procedure solves the output regulation problem.
With the initial condition $x(0) = [-0.149\;\; 0.2225\;\; 0.7115\;\; 0.3416]^\top$, we obtain the gains
\begin{equation}
    \begin{split}
        K_{\zeta}\! &= \! \begin{bsmallmatrix*}[r]
            \!135.73&\;  -28.239&\;   37.946&\;   89.622&\; -338.361&\;  169.546&\;   10.634&  -47.635\!\\
            \!12.709&\;    8.242&\;   18.999&\;   29.492&\; -116.075&\;  111.063&\;   -3.653&\;    3.838\!
        \end{bsmallmatrix*}\\
        K_{\eta}\!&=\!\begin{bsmallmatrix*}[r]
            8.09&\;   -4.716\\
            2.241&\;    7.497
        \end{bsmallmatrix*},
    \end{split}
\end{equation}
which place the eigenvalues of the interconnection of the plant and \eqref{eq:controller_out_reg} in $\{-1.199, -2.238 \pm 4.018\mathrm{i}, -2.434, -2.487 \pm 1.838\mathrm{i}, -2.782 \pm 96.497\mathrm{i}, -4, -4, -4.701, -7.567, -8, -8\}$.
We finally show the regulation performance of the controller in the closed-loop simulation run of Fig. \ref{fig:batch_reactor}.
In particular, we feed the controller with $e \coloneqq y_p - y^\star$, where $y_p \coloneqq Cx = \col(y_{p1}, y_{p2})$ is the plant output and $y^\star \coloneqq \col(y_1^\star, y_2^\star)$ is a piecewise-constant reference.
\begin{figure}[t!]
    \centering
    \includegraphics[width = \columnwidth]{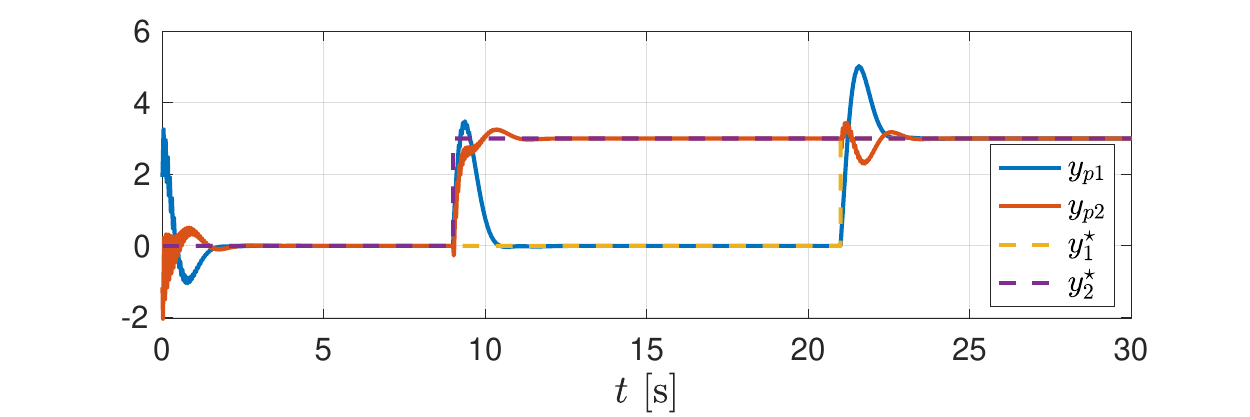}

    \vspace{-5pt}
    
    \caption{Simulation run for the batch reactor (control with integral action).}\label{fig:batch_reactor}

    \vspace{-10pt}
    
\end{figure}

\subsection{Motion Control of a Surface Vessel}
As a second example, we address the output regulation problem for the surface vessel dynamics given in \cite[scenario B]{pyrkin2019adaptive}.
The matrices of system \eqref{eq:plant_w} are
\begin{equation}
    \begin{split}
        A\! &=\! \begin{bmatrix*}[r]
            -0.1 &\; 0.012 &\; 0.015 &\;\;\;\; 0 &\;\;\; 0 &\; 0.01\\
            0.01 &\; -0.0333 &\; -0.05 &\;\;\; \;0 &\;\;\; 0 &\; -0.014\\
            0.02 &\; 0.03 & -0.18 &\;\;\;\; 0 &\;\;\; 0 &\; 0\\
            1 &\; 0 &\; 0 &\;\;\;\; 0 &\;\;\; 0 &\; 0\\
            0 &\; 1 &\; 0 &\;\;\;\; 0 &\;\;\; 0 &\; 0\\
            0 &\; 0 &\; 1 &\;\;\;\; 0 &\;\;\; 0 &\; 0
        \end{bmatrix*}\\
        B\! &=\! \begin{bmatrix*}[r]
            0 & 0.03 & 0.025\\
            0 & 0.21 & -0.2\\
            0.1 & 0.03 & 0.02\\
            0 & 0 & 0\\
            0 & 0 & 0\\
            0 & 0 & 0
        \end{bmatrix*}\!,\; P\! =\! \begin{bmatrix*}[r]
            -0.001 & 0 & 0.002\\
            0.02 & 0.01 & -0.02\\
            0 & 0 & 0\\
            0 & 0 & 0\\
            0.1 & 0 & 0\\
            0.1 & 0.1 & -0.1
        \end{bmatrix*}\\   
        C\! &=\!\begin{bmatrix*}[r]
            0 & 0 & 0 & 1 & 0 & 0\\
            0 & 0 & 0 & 0 & 1 & 0\\
            0 & 0 & 0 & 0 & 0 & 1
        \end{bmatrix*}\!, \qquad S = \begin{bmatrix}
            0 & 1 & 0\\
            0 & 0 & 1\\
            0 & -(\pi/5)^2 & 0
        \end{bmatrix}\!,
    \end{split}
\end{equation}
where $S$ generates a bias and a sinusoidal term at $\pi/5$ rad/s.
We also suppose that a bias affects the first output, namely:
\begin{equation}
    Q = 2\begin{bmatrix}
            1 & 0 & (5/\pi)^2\\
            0 & 0 & 0\\
            0 & 0 & 0
        \end{bmatrix}\!,
\end{equation}
and we let $e \in \R^2$ be the first two components of the output.
Note that Assumptions \ref{hyp:ctrb_obs}, \ref{hyp:non-resonance}, and \ref{hyp:nu} hold, with $\nu = 2$.

The dataset $(u(t), e(t), y_r(t))$ is obtained by simulating the plant for $\tau = 35$ s.
Each input channel is composed of $4$ distinct sinusoids.
We let $w(0) = [1\; 1\; 1]^\top$, while each component of $x(0)$ is extracted from the uniform distribution $U(-1, 1)$.

\begin{figure}[t!]
    \centering
    \includegraphics[width = \columnwidth]{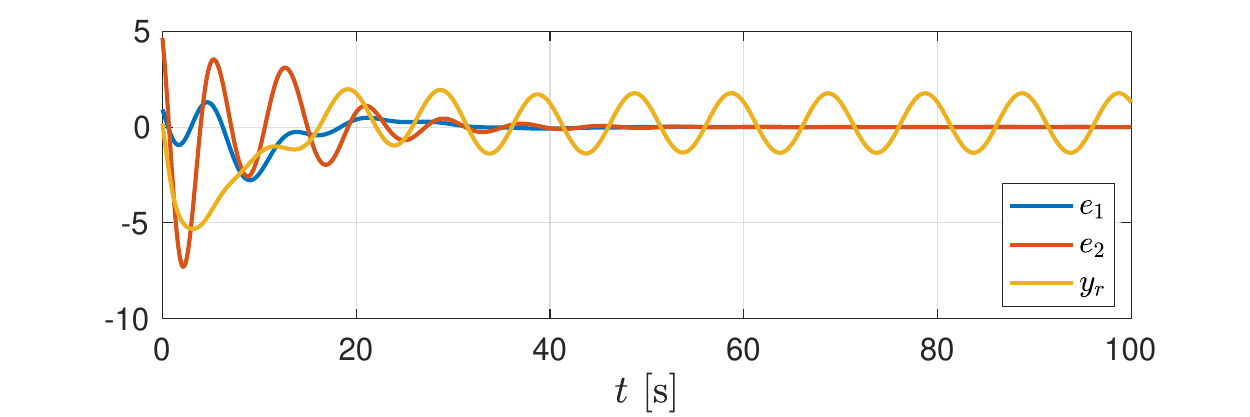}

    \vspace{-5pt}
    
    \caption{Simulation run for the surface vessel example.}\label{fig:surface_vessel}

    \vspace{-10pt}
    
\end{figure}

In this example, we assume that $\nu$ is already available.
Accordingly, we design a regulator of the form \eqref{eq:controller_out_reg} using Algorithm \ref{alg:out_reg}, with $N = 80$ samples and filter matrices $F$, $G$, and $L$ as in \eqref{eq:FGL_nu} with
\begin{equation}
    \Lambda = \begin{bmatrix*}[r]
        0 & 1\\
        -2 & -2
    \end{bmatrix*}, \qquad \ell = \begin{bmatrix}
        0 \\ 1/2
    \end{bmatrix},
\end{equation}
so that $\sigma(\Lambda) = \{-1 \pm \mathrm{i}\}$.
Also, the internal model is chosen with $\Phi$ and $\Gamma$ as in \eqref{eq:IM}, with $S_0 = S$ and $\Gamma_0 = [0\;\; 0\;\; 1/10]^\top$.
For each random initial condition $x(0)$, the resulting controller solves the output regulation.
For the dataset generated with $x(0) = [0.7297\;\; -0.7195\;\; 0.3143\;\; 0.186\;\; 0.0267\;\; -0.5108]^\top$, we obtain the following closed-loop system eigenvalues: $\{-0.11 \pm 0.2\mathrm{i}, -0.11 \pm 0.76\mathrm{i}, -0.15 \pm 0.79\mathrm{i}, -0.22 \pm 0.21\mathrm{i}, -0.27 \pm 0.53\mathrm{i}, -0.3, -0.31 \pm 0.16\mathrm{i}, -0.48 \pm 1.53\mathrm{i}, -0.86, -1 \pm 1\mathrm{i}, -1 \pm 1\mathrm{i}, -1 \pm 1\mathrm{i}, -1.56 \pm 316.19\mathrm{i}\}$.
We omit the control gains $K_{\zeta}$ and $K_{\eta}$, which can be found at the linked repository.
Also, the regulated output converges to zero as shown in Fig. \ref{fig:surface_vessel}, obtained by setting random initial conditions for the plant and the controller, and letting $w(0) = [1\;\; -3\;\; 0]^\top$.

\section{Conclusion}\label{sec:conclusion}
We presented a data-driven control framework for continuous-time LTI systems based on canonical non-minimal realizations.
The approach was applied to the stabilization and output regulation of MIMO systems.
We also proposed tuning strategies for the realizations under Assumption \ref{hyp:nu}, requiring a uniform observability index across all outputs.
Several research directions arise from this work.
First, future efforts will address robustness to noise, building on techniques such as \cite{van2020noisy, bisoffi2022data}.
While we make here the important assumption of noise-free data, we observe that this is a standard simplifying assumption in the emerging research field of continuous-time data-driven control \cite{lopez2024input, hu2025data, rapisarda2023orthogonal, ohta2024sampling}.
Another direction is to extend the framework to general MIMO systems by relaxing Assumption \ref{hyp:nu}.
In this regard, the methods in \cite{li2024controller, alsalti2025notes} suggest promising developments for the non-minimal realization framework.
Finally, it will be valuable to explore extensions of the framework to nonlinear systems.

\appendices

\section{Multivariable Observer Canonical Form}\label{sec:MIMO_obs_form}

Consider system \eqref{eq:plant} and suppose that $(C, A)$ is observable and $\rank C = p$.
Let $\nu_1, \ldots, \nu_p$ be the observability indices of the system and define $\nu_{ij} \coloneqq \min\{\nu_i, \nu_j\}$, for $i \neq j$.
Then, there exists a similarity transformation $T_o$ that puts \eqref{eq:plant} in the observer canonical form \cite[Ch. 3]{antsaklis1997linear}, \cite{denham1974canonical}:
\begin{equation}\label{eq:MIMO_obs_form}
    \begin{split}
        T_o A T_o^{-1} = A_o &\coloneqq \begin{bmatrix}
            A_{11} & \cdots & A_{1p}\\
            \vdots & & \vdots\\
            A_{p1} & \cdots & A_{pp}
        \end{bmatrix}, \quad T_o B = B_o\\
        C T_o^{-1} = C_o &\coloneqq \begin{bmatrix}
            C_1 & \cdots & C_p
        \end{bmatrix},
    \end{split}
\end{equation}
where $B_o$ has no particular form, while
\begin{equation}\label{eq:MIMO_obs_AC}
    \begin{split}
        A_{ii} \coloneqq& \left[\begin{array}{c|c}
            \; \begin{array}{c}
                0_{1\times(\nu_i - 1)} \\
                \hline
                I_{\nu_i - 1}
            \end{array}\;\,
                & \; \alpha_{ii} \;
            \end{array}\right] \! \in \! \R^{\nu_i \times \nu_i}\!, \quad i \in \{1, \ldots, p\}\\
            A_{ij} \coloneqq& \left[\begin{array}{c|c}
        \!\!\!0_{\nu_i \times (\nu_j - 1)}\!\!\! & \!\!\!\! \begin{array}{c}
                \alpha_{ij}\\
                0_{(\nu_i - \nu_{ij})\times 1}
            \end{array}\!\!\!\!\!
        \end{array}
    \right]\!\! \in \! \R^{\nu_i \times \nu_j}\!, \;\; i, j \! \in \! \{1, \ldots, p\}, i \neq j\\
    C_i \coloneqq&\left[\begin{array}{c}
         0_{(i-1)\times \nu_i}  \\
         \hline
         \begin{array}{c|c}
             0_{(p - i + 1)\times (\nu_i - 1)} & \begin{array}{c}
                  1\\
                  c_i 
             \end{array}\!\!\!
         \end{array}
    \end{array}
    \right]\! \in \! \R^{p \times \nu_i}\!, \; i \in \{1, \ldots, p\},
    \end{split}
\end{equation}
with $\alpha_{ii} \in \R^{\nu_i}$, $\alpha_{ij} \in \R^{\nu_{ij}}$ and $c_i \in \R^{p-i}$ depending on $(C, A)$.

Matrices $A_o$ and $C_o$ can also be written as
\begin{equation}\label{eq:AC_decomposition}
    A_o = \bar{A} + A_m\bar{C}, \qquad C_o = C_m\bar{C},
\end{equation}
where $A_m \in \R^{n \times p}$ and $C_m \in \R^{p \times p}$ contain the coefficients of $\alpha_{ii}$, $\alpha_{ij}$, and $c_i$, while $\bar{A}$ and $\bar{C}$ are defined as
\begin{equation}\label{eq:bar_AC}
    \begin{split}
        \bar{A} &\coloneqq \bdiag(\bar{A}_i, \ldots, \bar{A}_p), \qquad \bar{C} \coloneqq \bdiag(\bar{c}_i^\top, \ldots, \bar{c}_p^\top)\\
        \bar{A}_i &\coloneqq \begin{bmatrix}
            0_{1 \times (\nu_i-1)} & 0 \\
            I_{\nu_i-1} & 0_{(\nu_i-1)\times 1}
        \end{bmatrix}\!, \quad\; \bar{c}_i^\top \coloneqq \begin{bmatrix}
            0_{1\times (\nu_i - 1)} \!& 1
        \end{bmatrix}\!.
    \end{split}
\end{equation}
Finally, if the observability indices satisfy $\nu_1 \leq \nu_2 \leq \ldots \leq \nu_p$, it holds that $C_m = I_p$, thus $C_o = \bar{C}$ \cite[Ch. 3, Lemma 4.9]{antsaklis1997linear}.

\section{A Useful Linear Matrix Equation}\label{sec:X}
Some technical proofs of this article rely on this result.
\begin{lemma}\label{lem:X}
    Consider any matrix $\Theta \in \R^{r \times r}$ and any vectors $\beta \in \R^{r}$ and $\phi \in \R^{r}$ such that the pair $(\Theta, \beta)$ is controllable.
    Then, there exists a unique solution $X \in \R^{r \times r}$ to the equation:
    \begin{equation}\label{eq:X}
        \Theta X = X \Theta, \qquad X\beta = \phi.
    \end{equation}
\end{lemma}
\begin{proof}
    We use the arguments of \cite[Prop. 4.1]{serrani2000semiglobal} to compute an explicit solution to \eqref{eq:X}.
    Since $(\Theta, \beta)$ is controllable, $\Theta$ is cyclic, i.e., its minimal and characteristic polynomials coincide.
    Then, by \cite[Pag. $222$]{gantmakher2000theory}, any matrix $X$ that satisfies equation $X\Theta = \Theta X$ of \eqref{eq:X}, can be expressed as
    \begin{equation}\label{eq:polynomial}
        X = \sum_{i = 0}^{r - 1}\varrho_{i} \Theta^i,
    \end{equation}
    with free parameters $\varrho_{0}, \ldots, \varrho_{r-1}$.
    By replacing \eqref{eq:polynomial} in the second equation of \eqref{eq:X}, it holds that
    \begin{equation}
        \sum_{i = 0}^{r - 1}\varrho_{i} \Theta^i \beta = R\varrho = \phi,
    \end{equation}
    where $R \coloneqq [\beta \; \Theta\beta \; \cdots \; \Theta^{r - 1}\beta]$ and $\varrho \coloneqq \col(\varrho_{0}, \ldots, \varrho_{r-1})$.
    Since $(\Theta, \beta)$ is controllable, $R$ is invertible and we obtain $\varrho = R^{-1}\phi$, so $X$ is uniquely determined.
\end{proof}

\section{Technical Proofs}\label{sec:proofs}

\subsection{Proof of Lemma \ref{lem:min_poly}}
Rewrite the first equation of \eqref{eq:Pi} as $\Pi F = (A - \Pi L C)\Pi$, then note that, post-multiplying both sides by $F$, we obtain
\begin{equation}
    \Pi F^2 = (A - \Pi L C) \Pi F = (A - \Pi L C)^2 \Pi.
\end{equation}
This operation can be repeated for arbitrary powers.
Therefore, for any polynomial $\rho(\cdot)$, the following identity holds:
\begin{equation}\label{eq:Pi_poly}
    \Pi \rho(F) = \rho(A - \Pi L C) \Pi.
\end{equation}
Let $v \in \C^n$ be any generalized left eigenvector of $A - \Pi L C$ with associated eigenvalue $s$, i.e., $v \neq 0$ such that, for some $i \in \N$ with $i \geq 1$, $v^\top (A - \Pi L C - s I_n)^i = 0$ and $v^\top (A - \Pi L C - s I_n)^j \neq 0$, for all $j \in \{0, \ldots, i-1\}$.
Using \eqref{eq:Pi_poly}, we obtain
\begin{equation}\label{eq:Pi_powers}
    \Pi(F - s I_\mu)^j = (A - \Pi L C - s I_n)^j\Pi, \quad j \in \{1, \ldots, i\}.
\end{equation}
Note that, due to $\rank \Pi = n$, any vector $\phi \in \C^n$, $\phi \neq 0$, is such that $\Pi^\top \phi \neq 0$.
Then, pre-multiplying both sides of \eqref{eq:Pi_powers} by the generalized left eigenvector $v^\top$, it follows that
\begin{equation}
    \begin{split}
        v^\top \Pi(F - s I_\mu)^j \neq 0,& \quad j \in \{1, \ldots, i - 1\}\\
        v^\top \Pi(F - s I_\mu)^i = 0,&
    \end{split}
\end{equation}
where $\Pi^\top v \neq 0$.
As a consequence, $\Pi^\top v$ is a generalized left eigenvector of $F$, with associated eigenvalue $s$.
Also, if some vectors $v_1 \neq 0$ and $v_2 \neq 0$ are linearly independent, then so are $\Pi^\top v_1$ and $\Pi^\top v_2$ because $\Pi^\top (v_1 - c v_2) = 0$, for some $c$, would contradict $\rank \Pi = n$.
Since we can find $n$ linearly independent generalized eigenvectors of $A - \Pi L C$, we can find $n$ corresponding linearly independent generalized eigenvectors of $F$ having the same eigenvalues.
This matching property between the eigenspaces implies the statement.
\hfill\QED

\subsection{Proof of Lemma \ref{lem:chi}}
By Lemma \ref{lem:min_poly} and the construction of $F_0$ in \eqref{eq:Delta}, all the modes of $A - \Pi L C$ are also modes of $F_0$.
Therefore, for any $\epsilon(0) = x(0)$, there exists a matrix $E \in \R^{\mu \times \delta}$ and a vector $\tilde{\epsilon}_0 \in \R^{\delta}$ such that $LC\epsilon(0) = E\tilde{\epsilon}_0$ and, for all $t\in[0, \tau]$,
\begin{equation}
    LC \epsilon(t) = Ee^{F_0 t}\tilde{\epsilon}_0 = E\tilde{\epsilon}(t),
\end{equation}
where $\tilde{\epsilon}(t) \in \R^\delta$ satisfies $\dot{\tilde{\epsilon}}(t) = F_0 \tilde{\epsilon}(t)$ and $\tilde{\epsilon}(0) = \tilde{\epsilon}_0$.
Since $(F_0, G_0)$ is controllable, it follows that $\tilde{\epsilon}(t) = Y\chi(t)$, where $Y$ is computed from
\begin{equation}
    Y F_0 = F_0 Y, \quad YG_0 = \tilde{\epsilon}_0,
\end{equation}
whose solution exists and is unique due to Lemma \ref{lem:X} given in Appendix \ref{sec:X}.
The proof is concluded by letting $D = EY$.
\hfill\QED

\subsection{Proof of Theorem \ref{thm:tuning_nu}}
\subsubsection{From a quadratic to a linear equation}
By rearranging the first equation of \eqref{eq:Pi} and using $H = C\Pi$, we obtain the following quadratic equation
\begin{equation}\label{eq:Pi_reduced}
    \Pi F = (A - \Pi L C)\Pi, \qquad \Pi G = B.
\end{equation}
Thus, we can find $\Pi$ with \eqref{eq:Pi_reduced} and then compute $H = C\Pi$.

Since pair $(C, A)$ is observable by Assumption \ref{hyp:ctrb_obs}, there exists a non-singular matrix $T_o \in \R^{n \times n}$ such that the matrices $A_o \coloneqq T_o A T_o^{-1}$, $B_o \coloneqq T_o B$, $C_o \coloneqq CT_o^{-1}$ are in observer canonical form, see Appendix \ref{sec:MIMO_obs_form}.
In particular, from \eqref{eq:AC_decomposition}, $A_o = \bar{A} + A_m \bar{C}$.
Also, since all observability indices are equal, $C_o = \bar{C}$.
Defining $Y \coloneqq T_o\Pi$, we can transform \eqref{eq:Pi_reduced} into:
\begin{equation}\label{eq:Y_quad}
    Y F = (\bar{A} + (A_m - Y L) \bar{C})Y, \quad Y G = B_o.
\end{equation}
Let
\begin{equation}\label{eq:D_lambda}
    \mathcal{D}_\lambda(s) \coloneqq \det(sI_\nu - \Lambda) = s^\nu + \theta_{\lambda, \nu -1}s^{\nu-1} + \ldots + \theta_{\lambda, 0},
\end{equation}
then define
\begin{equation}\label{eq:Lambda_o}
    \Lambda_o \coloneqq  \begin{bmatrix}
        0 & \cdots & 0 & -\theta_{\lambda, 0}\\
        1 & \cdots & 0 & -\theta_{\lambda, 1}\\
        \vdots & \ddots & \vdots & \vdots\\
        0 & \cdots & 1 & -\theta_{\lambda, \nu-1}
    \end{bmatrix}\!.
\end{equation}
Also, let $\theta_{\lambda} \coloneqq \col(\theta_{\lambda, 0}, \ldots, \theta_{\lambda, \nu - 1})$ and $\Psi \coloneqq A_m + I_p \otimes \theta_\lambda \in \R^{n \times p}$.
Recalling the structure of $\bar{A}$ and $\bar{C}$ in \eqref{eq:bar_AC}, we obtain
\begin{equation}\label{eq:Psi_assignment}
    \bar{A} + (A_m - \Psi) \bar{C} = I_p \otimes \Lambda_o.
\end{equation}
To solve the quadratic equation \eqref{eq:Y_quad}, we introduce the additional constraint $Y L = \Psi$, so that from \eqref{eq:Psi_assignment} we obtain
\begin{equation}\label{eq:Y_lin}
    Y F = (I_p \otimes \Lambda_o)Y, \qquad Y G = B_o, \qquad Y L = \Psi.
\end{equation}
Any solution to \eqref{eq:Y_lin} is also a solution to \eqref{eq:Y_quad} and ensures that $I_p \otimes \Lambda$ and $A - \Pi L C = T_o^{-1}(I_p \otimes \Lambda_o) T_o$ are similar.
\subsubsection{Equation decoupling}
The rest of the proof involves exploiting the tuning \eqref{eq:FGL_nu} to decouple \eqref{eq:Y_lin} into blocks that can be solved with Lemma \ref{lem:X}.
Let $Y = [Y_y\; Y_u]$, with $Y_y \in \R^{n \times n}$ and $Y_u \in \R^{n \times m\nu}$.
Then, we can decouple \eqref{eq:Y_lin} into
\begin{subequations}
    \begin{eqnarray}
        \label{eq:Y_y}
        Y_y (I_p \otimes \Lambda) = (I_p \otimes \Lambda_o)Y_y,& \quad Y_y (I_p \otimes \ell) = \Psi,\\
        \label{eq:Y_u}
        Y_u (I_m \otimes \Lambda) = (I_p \otimes \Lambda_o)Y_u,& \quad Y_u (I_m \otimes \ell) = B_o.
    \end{eqnarray}
\end{subequations}
By letting $Y_y = [Y_{y, ij}]_{1 \leq i, j \leq p}$, $Y_{u} = [Y_{u, ij}]_{1 \leq i \leq p, 1 \leq j \leq m}$, with blocks $Y_{y, ij}$, $Y_{u, ij} \in \R^{\nu \times \nu}$, and $\Psi = [\psi_{ij}]_{1 \leq i, j \leq p}$, $B_o = [b_{ij}]_{1 \leq i \leq p, 1 \leq j \leq m}$, with vectors $\psi_{ij}$, $b_{ij} \in \R^\nu$, \eqref{eq:Y_y} and \eqref{eq:Y_u} can be further split into equations of the form
\begin{subequations}
    \begin{eqnarray}
        Y_{y, ij} \Lambda = \Lambda_{o} Y_{y, ij},& \qquad Y_{y, ij} \ell = \psi_{ij},\\
        Y_{u, ij} \Lambda = \Lambda_{o} Y_{u, ij},& \qquad Y_{u, ij} \ell = b_{ij}.
    \end{eqnarray}
\end{subequations}
Finally, by letting $T_\lambda$ be the similarity transformation such that $\Lambda = T_\lambda\Lambda_o T_\lambda^{-1}$, we obtain
\begin{subequations}\label{eq:X_ij}
    \begin{eqnarray}
        X_{y, ij} \Lambda \! = \! \Lambda X_{y, ij},& \!\!\!X_{y, ij} \ell \! = \! T_\lambda \psi_{ij}, \;\;\;\; i, j \! \in \! \{1, \ldots, p\}\\
        X_{u, ij} \Lambda \! = \! \Lambda X_{u, ij},& X_{u, ij} \ell \! = \! T_\lambda b_{ij}, \;\;\;\; \left\{\!\!\!\begin{array}{c}
            i \! \in \! \{1, \ldots, p\}\\
            j \! \in \! \{1, \ldots, m\},
        \end{array}\right.
    \end{eqnarray}
\end{subequations}
where $X_{y, ij} \coloneqq T_\lambda Y_{y, ij}$, $X_{u, ij} \coloneqq T_\lambda Y_{u, ij}$.
We conclude the proof by invoking Lemma \ref{lem:X} to solve each equation in \eqref{eq:X_ij}.
\hfill\QED

\subsection{Proof of Theorem \ref{thm:ctrb_nu}}
\subsubsection{Input-output model of the non-minimal realization}
Let $V$ be the similarity transformation such that pair $(\Lambda_c, \ell_c) \coloneqq (V \Lambda V^{-1}, V \ell)$ is in controller canonical form:
\begin{equation}\label{eq:Lambda_c}
    \Lambda_c \coloneqq \begin{bmatrix}
        0 & 1 & \cdots & 0\\
        \vdots & \vdots & \ddots & \vdots\\
        0 & 0 & \cdots & 1\\
        -\theta_{\lambda, 0} & -\theta_{\lambda, 1} & \cdots & -\theta_{\lambda, \nu - 1}
    \end{bmatrix}\!, \quad \ell_c \coloneqq \begin{bmatrix}
        0 \\ \vdots \\ 0 \\  1
        \end{bmatrix}\!,
\end{equation}
where we used the parameters $\theta_\lambda$ defined in \eqref{eq:D_lambda}.
Then, using the definition of $F$, $G$, and $L$ in \eqref{eq:FGL_nu}, the pair $((I_{p + m} \otimes V)(F + LH)(I_{p + m} \otimes V)^{-1}, (I_{p + m} \otimes V)G)$ can be written as:
\begin{equation}
    \left(\begin{bmatrix}
        I_p \otimes \Lambda_c + (I_p \otimes \ell_c)\Theta_y & (I_p \otimes \ell_c)\Theta_u\\
        0_{m\nu \times n} & I_m \otimes \Lambda_c
    \end{bmatrix}, \begin{bmatrix}
        0_{n \times m}\\
        I_m \otimes \ell_c
    \end{bmatrix}\right),
\end{equation}
where $\Theta_y \in \R^{p \times n}$ and $\Theta_u \in \R^{p \times m\nu}$ are given by $\begin{bmatrix}
        \Theta_y & \Theta_u
    \end{bmatrix} \coloneqq H (I_{p + m}\otimes V^{-1})$.
In the coordinates $\zeta_c \coloneqq (I_{p + m} \otimes V)\zeta$, system \eqref{eq:non-minimal_plant_filter} becomes
\begin{equation}
    \begin{split}
        \dot{\zeta}_c &= \begin{bmatrix}
            I_p \otimes \Lambda_c & 0_{n \times m\nu}\\
            0_{m\nu \times n} & I_m \otimes \Lambda_c
        \end{bmatrix}\zeta_c + \begin{bmatrix}
            I_p \otimes \ell_c & 0_{n \times m}\\
            0_{m\nu \times p} & I_m \otimes \ell_c
        \end{bmatrix}\begin{bmatrix}
            y \\ u
        \end{bmatrix}\\
        y &= \begin{bmatrix}
        \Theta_y & \Theta_u
    \end{bmatrix} \zeta_c.
    \end{split}
\end{equation}
Applying the Laplace transform and the properties of the Kronecker product to the above equations, we obtain
\begin{equation}
    \begin{split}
        y(s) =\;& \Theta_y(I_p \otimes (sI_\nu - \Lambda_c)^{-1}\ell_c)y(s)\; +\\
        &+ \Theta_u(I_m \otimes (sI_\nu - \Lambda_c)^{-1}\ell_c)u(s),\\
    \end{split}
\end{equation}
where $u(s)$ and $y(s)$ are the transforms of $u(t)$ and $y(t)$.

Then, exploiting the controllable version of the structure theorem \cite[Ch. 3, Thm. 4.10]{antsaklis1997linear}, we obtain
\begin{equation}
    \begin{split}
        y(s) &= \frac{1}{\mathcal{D}_\lambda(s)}\begin{bmatrix}
            \Theta_y & \Theta_u
        \end{bmatrix} \left(I_{p + m} \otimes \begin{bmatrix}
            1 \\ \vdots \\ s^{\nu-1}
        \end{bmatrix}\right)\begin{bmatrix}
            y(s)\\ u(s)
        \end{bmatrix}\\
        &=\frac{\mathcal{N}_y(s)}{\mathcal{D}_\lambda(s)}y(s) + \frac{\mathcal{N}_u(s)}{\mathcal{D}_\lambda(s)}u(s),
    \end{split}
\end{equation}
where $\mathcal{D}_\lambda(s) \in \C$ is defined in \eqref{eq:D_lambda}, while $\mathcal{N}_y(s) \in \C^{p \times p}$ and $\mathcal{N}_u(s) \in \C^{p \times m}$ are polynomial matrices.

Using $s$ with some abuse of notation to also denote the differential operator $d/dt(\cdot)$, system \eqref{eq:non-minimal_plant} can be represented with the following polynomial matrix description:
\begin{equation}
    (\mathcal{D}_\lambda(s) I_p - \mathcal{N}_y(s))y(t) = \mathcal{N}_u(s)u(t).
\end{equation}
See \cite[Ch. 7]{antsaklis1997linear} for definitions and results related to polynomial matrix descriptions of LTI systems.
Note that the polynomial matrix $\mathcal{D}_\lambda(s) I_p - \mathcal{N}_y(s)$ is such that the matrix of coefficients of degree $\nu$ is $I_p$ since each entry of $\mathcal{N}_y(s)$ is a polynomial of degree at most $\nu  - 1$.
As a consequence, $\deg \det(\mathcal{D}_\lambda(s) I_p - \mathcal{N}_y(s)) = n$.
On the other hand, exploiting the observable version of the structure theorem \cite[Ch. 3, Thm. 4.11]{antsaklis1997linear}, there exists a polynomial matrix description of plant \eqref{eq:plant} of the form
\begin{equation}\label{eq:obs_structure}
    \mathcal{D}(s)y(t) = \mathcal{N}(s)u(t),
\end{equation}
where $\mathcal{D}(s)$ and $\mathcal{N}(s)$ are left coprime because \eqref{eq:obs_structure} is equivalent to system \eqref{eq:plant}, which is controllable by Assumption \ref{hyp:ctrb_obs} (see \cite[Ch. 7, Pag. 560 and Thm. 3.4]{antsaklis1997linear} for further details).
Also, by construction as in \cite[Ch. 3, Thm. 4.11]{antsaklis1997linear} and Assumption \ref{hyp:nu}, the matrix of coefficients of degree $\nu$ of $\mathcal{D}(s)$ is $I_p$, thus $\deg\det(\mathcal{D}(s)) = n$.
    
Moreover, by Lemma \ref{lem:Pi} and Remark \ref{rem:transfer}, systems \eqref{eq:plant} and \eqref{eq:non-minimal_plant} have the same transfer matrix, i.e., $\mathcal{D}^{-1}(s)\mathcal{N}(s)= (\mathcal{D}_\lambda(s) I_p - \mathcal{N}_y(s))^{-1}\mathcal{N}_u(s)$.
Then, there exists a matrix $\mathcal{R}(s)$ such that
\begin{equation}\label{eq:i_o_matching}
    \begin{bmatrix}
        \mathcal{D}_\lambda(s) I_p - \mathcal{N}_y(s) & \mathcal{N}_u(s)
    \end{bmatrix} = \mathcal{R}(s)\begin{bmatrix}
        \mathcal{D}(s) & \mathcal{N}(s) 
    \end{bmatrix}.
\end{equation}
However, from $\det(\mathcal{D}_\lambda(s) I_p - \mathcal{N}_y(s)) = \det(\mathcal{R}(s)\det(\mathcal{D}(s))$ and $\deg \det(\mathcal{D}_\lambda(s) I_p - \mathcal{N}_y(s)) = \deg\det(\mathcal{D}(s)) = n$, we obtain $\deg\det(\mathcal{R}(s)) = 0$ and, as a consequence, $\det(\mathcal{R}(s)) = \alpha \neq 0$, for some $\alpha \in \R$.
In other words, $\mathcal{R}(s)$ is unimodular.
Thus, by \cite[Ch. 7, Thm 2.5]{antsaklis1997linear}, we conclude that
\begin{equation}\label{eq:rank:D_N}
    \rank\begin{bmatrix}
        \mathcal{D}_\lambda(s) I_p - \mathcal{N}_y(s) & \mathcal{N}_u(s)
    \end{bmatrix} = p.
\end{equation}

\subsubsection{PBH test} We now prove the statement exploiting \eqref{eq:rank:D_N}.
In particular, using the PBH test, we verify that
\begin{equation}\label{eq:PBH_nu}
    \begin{bmatrix}
        sI_n - I_p \! \otimes \! \Lambda_c - (I_p \! \otimes \ell_c)\Theta_y & -(I_p \! \otimes \ell_c)\Theta_u & 0_{n \times m}\\
        0_{m\nu \times n} & sI_{m\nu} - I_m \! \otimes \! \Lambda_c & I_m \! \otimes \ell_c
    \end{bmatrix}
\end{equation}
has rank $\mu = n + m\nu$, for all $s \in \sigma (I_p \otimes \Lambda_c + (I_p \otimes \ell_c)\Theta_y) \cup \sigma (\Lambda)$.
Define the unimodular matrix
\begin{equation}
    \mathcal{T} \coloneqq \begin{bmatrix}
        1 & 0 &  & \cdots & 0\\
        s & 1 & 0 & & \vdots\\
        s^2 & s & 1 & \ddots& \\
        \vdots & \vdots & \ddots  & \ddots & 0\\
        s^{\nu-1} & s^{\nu-2} & \cdots & s & 1
    \end{bmatrix} \in \C^{\nu \times \nu},
\end{equation}
then post-multiply \eqref{eq:PBH_nu} by matrix $\bdiag(I_{p + m} \otimes \mathcal{T}, I_m) \in \C^{(\mu + m)\times(\mu + m)}$.
This operation is justified by the fact that
\begin{equation}
    (sI_\nu - \Lambda_c) \mathcal{T}= \begin{bmatrix}
        0 & -1 & \cdots & 0\\
        \vdots & \vdots & \ddots & \vdots\\
        0 & 0 & \cdots & -1\\
        \mathcal{D}_{\lambda}(s) & \star & \cdots & \star
    \end{bmatrix},
\end{equation}
where $\star$ denotes polynomials of $s$, and similarly, for any $\theta \coloneqq \col(\theta_0, \ldots, \theta_{\nu-1})$,
\begin{equation}
    -\ell_c\theta^\top \mathcal{T} = \begin{bmatrix}
        0 & 0 & \cdots & 0\\
        \vdots & \vdots & & \vdots\\
        0 & 0 & \cdots & 0\\
        -\mathcal{N}_\theta(s) & \star & \cdots & \star
    \end{bmatrix},
\end{equation}
where $\mathcal{N}_\theta(s) \coloneqq \theta_{\nu-1}s^{\nu - 1} + \ldots + \theta_0$.
A similar operation is performed in the proof of \cite[Lem. 4.2]{isidori2017lectures}.
After applying a simple permutation of the rows to group the polynomials in the last row of each $\nu \times \nu$ block, we obtain the following matrix 
(where blank terms represent zeros, $\star$ again indicate polynomials of $s$, and lines highlights the same blocks appearing in \eqref{eq:PBH_nu}):
\begin{equation}
    \left[\begin{array}{cc|cc|c}
        & -I_{n-p} & & & \\
        \mathcal{D}_{\lambda}(s)I_p - \mathcal{N}_y(s) & \star & -\mathcal{N}_u(s)  & \star & \\
        \hline
        & & & -I_{m(\nu - 1)} & \\
        & & \mathcal{D}_{\lambda}(s)I_{m} & \star & I_m
    \end{array}\right].
\end{equation}
By performing row and column operations with unimodular matrices, we can use the identity matrices to eliminate the $\star$ terms and $\mathcal{D}_{\lambda}(s)I_m$ in the last block row.
After reordering the rows and the columns, we obtain that the PBH test with \eqref{eq:PBH_nu} is equivalent to verifying
\begin{equation}
    \rank\begin{bmatrix}
        I_{\mu - p} & 0_{(\mu - p) \times p} &0_{(\mu - p) \times m}\\
        0_{p \times (\mu - p)} & \mathcal{D}_\lambda(s) I_p - \mathcal{N}_y(s) & \mathcal{N}_u(s)
    \end{bmatrix} = \mu,
\end{equation}
which is true for all $s \in \C$ due to \eqref{eq:rank:D_N}.
\hfill\QED

\subsection{Proof of Theorem \ref{thm:nu_hat}}
Suppose that $\hat{\nu} \geq \nu + 1$, then define the matrix gains $\Lambda \coloneqq -\diag(\lambda_1, \ldots, \lambda_{\nu})$, $\tilde{\Lambda} \coloneqq -\diag(\lambda_{\nu + 1}, \ldots, \lambda_{\hat{\nu}})$, $\ell \coloneqq \col(\gamma_1, \ldots, \gamma_\nu)$, $\tilde{\ell} \coloneqq (\gamma_{\nu + 1}, \ldots, \gamma_{\hat{\nu}})$, and the dimensions $\mu_0 \coloneqq \nu(p + m)$, $\tilde{\mu} \coloneqq (\hat{\nu} - \nu)(p + m)$.
Note that the dynamics of $\hat{\zeta}$ in \eqref{eq:filter_nu_hat} can be decoupled into
\begin{equation}\label{eq:zeta_0_tilde}
    \begin{split}
        \dot{\hat{\zeta}}_0(t) &= (I_{p+m} \otimes \Lambda)\hat{\zeta}_0(t) + (I_{p+m} \otimes \ell)\begin{bmatrix}
            y(t) \\ u(t)
        \end{bmatrix}\\
        &= F\hat{\zeta}_0(t) + Gu(t) + Ly(t)\\
        \dot{\tilde{\zeta}}(t) &= (I_{p+m} \otimes \tilde{\Lambda})\tilde{\zeta}(t) + (I_{p+m} \otimes \tilde{\ell})\begin{bmatrix}
            y(t) \\ u(t)
        \end{bmatrix}\\
        &= \tilde{F}\tilde{\zeta}(t) + \tilde{G}u(t) + \tilde{L}y(t),
    \end{split}
\end{equation}
with $\hat{\zeta}_0 \in \R^{\mu_0}$ and $\tilde{\zeta} \in \R^{\tilde{\mu}}$.
By Theorem \ref{thm:tuning_nu}, there exist $\Pi \in \R^{n \times \mu_0}$ and $H \in \R^{p \times \mu_0}$ such that
\begin{itemize}
    \item \eqref{eq:Pi} holds, with $F$, $G$, and $L$ given in \eqref{eq:zeta_0_tilde}.
    \item $A - \Pi L C$ is similar to $I_p \otimes \Lambda$.
\end{itemize}
Then, applying Lemma \ref{lem:redundancy}, we obtain that $\hat{\Pi} \coloneqq [\Pi\;\; 0_{n \times \tilde{\mu}}]$ and $\hat{H} \coloneqq C\hat{\Pi}$ solve \eqref{eq:Pi}, with $F$, $G$, and $L$ replaced by
\begin{equation}
    \hat{F} \coloneqq \bdiag(F, \tilde{F}), \quad \hat{G} \coloneqq \begin{bmatrix}
        G \\ \tilde{G}
    \end{bmatrix}, \quad \hat{L} \coloneqq \begin{bmatrix}
        L \\ \tilde{L}
    \end{bmatrix}.
\end{equation}
In particular, note that
\begin{equation}
    \hat{F} + \hat{L}\hat{H} = \begin{bmatrix}
        F + LH & 0_{\mu\times\tilde{\mu}}\\
        \tilde{L}H & I_{p + m}\otimes \tilde{\Lambda}
    \end{bmatrix}.
\end{equation}
We can show that the pair $(\hat{F} + \hat{L}\hat{H}, \hat{G})$ is not controllable.
Indeed, the matrix
\begin{equation}\label{eq:PBH_nu_hat}
    \begin{bmatrix}
        sI_\mu - (F + LH) & 0_{\mu\times\tilde{\mu}} & G\\
        -\tilde{L}H & sI_{\tilde{\mu}} - (I_{p + m}\otimes \tilde{\Lambda}) & \tilde{G}
    \end{bmatrix}
\end{equation}
is such that, for any $s \in \sigma(\tilde{\Lambda})$, $p + m$ columns become zero due to the term $sI_{\tilde{\mu}} - (I_{p + m}\otimes \tilde{\Lambda})$.
As a consequence, for $s \in \sigma(\tilde{\Lambda})$, \eqref{eq:PBH_nu_hat} has $\hat{\nu}(p + m)$ rows and at most $(\hat{\nu} - 1)(p + m) + m = \hat{\nu}(p + m) - p < \hat{\nu}(p + m)$ linearly independent columns.
We infer that each eigenvalue of $\tilde{\Lambda}$ has an eigenvector belonging to the uncontrollable subspace.

Combining the previous results and following the procedure of Section \ref{sec:stabilization}, we obtain that the data $\chi(t)$, $\hat{\zeta}(t)$ in \eqref{eq:filter_nu_hat} satisfy a differential equation of the form \eqref{eq:chi_zeta}, which can be written explicitly as follows (where blank terms represent zeros):
\begin{equation}\label{eq:chi_zeta_nu_hat}
    \begin{bmatrix}
        \dot{\chi}_0(t)\\
        \dot{\tilde{\chi}}(t)\\
        \dot{\hat{\zeta}}_0(t)\\
        \dot{\tilde{\zeta}}(t)
    \end{bmatrix} = \begin{bmatrix}
        \Lambda &  &  & \\
         & \tilde{\Lambda} &  & \\
        D &  & F + LH & \\
        \tilde{D} &  & \tilde{L}H & I_{p + m} \otimes \tilde{\Lambda}
    \end{bmatrix}\begin{bmatrix}
        \chi_0(t)\\
        \tilde{\chi}(t)\\
        \hat{\zeta}_0(t)\\
        \tilde{\zeta}(t)
    \end{bmatrix} + \begin{bmatrix}
         \\ \\ G \\ \tilde{G}
    \end{bmatrix} u(t),
\end{equation}
with $\col(\chi_0(0), \tilde{\chi}(0)) = \col(\ell, \tilde{\ell})$ and $\col(\hat{\zeta}_0(0), \tilde{\zeta}(0)) = 0$.
Note that there is no coupling between $\tilde{\chi}(t)$ and $\col(\hat{\zeta}_0(t), \tilde{\zeta}(t))$ because \eqref{eq:chi_zeta_nu_hat} is obtained by applying Lemma \ref{lem:chi} to transform $\epsilon(t)$ into $\chi(t)$, and no modes of $\tilde{\Lambda}$ are contained in $\epsilon(t)$ since $A - \Pi L C$ is similar to $I_p \otimes \Lambda$.

The uncontrollable subsystem of \eqref{eq:chi_zeta_nu_hat} can be written as
\begin{equation}
    \begin{bmatrix}
        \dot{\tilde{\chi}}(t)\\
        \dot{\chi}_0(t)\\
        \dot{\hat{x}}_{\bar{c}}
    \end{bmatrix} = \begin{bmatrix}
        \tilde{\Lambda} &  & \\
         & \Lambda & \\
         & D_{\bar{c}} & A_{\bar{c}}
    \end{bmatrix}\begin{bmatrix}
        \tilde{\chi}(t)\\
        \chi_0(t)\\
        \hat{x}_{\bar{c}}(t)
    \end{bmatrix},
\end{equation}
for some matrices $D_{\bar{c}}$ and $A_{\bar{c}}$.
Note that $\sigma(\tilde{\Lambda}) \subset \sigma(A_{\bar{c}})$ as we have proved above using \eqref{eq:PBH_nu_hat}.
We conclude the proof by noticing that each eigenvalue $-\lambda_i$, $i \in \{\nu + 1, \ldots, \hat{\nu}\}$ appears both in $\sigma(\tilde{\Lambda})$ and in $\sigma\Bigl(\begin{bsmallmatrix}
    \Lambda & \\
    D_{\bar{c}} & A_{\bar{c}}
\end{bsmallmatrix}\Bigr)$.
Thus, $\mathcal{B}$ in \eqref{eq:batches_nu_hat} loses rank due to the fact that, after a change of coordinates, two rows of the data are equal to $[1\; e^{-\lambda_i \Ts}\; \cdots e^{-\lambda_i(N-1)\Ts}]$.
\hfill\QED

\section*{Acknowledgment}
We thank Prof. Andrea Serrani for the insightful discussions inspiring the results of this article.

\section*{References}

\vspace{-10pt}

\bibliographystyle{IEEEtran}
\bibliography{data-driven_bib}

\begin{thebibliography}{10}
\providecommand{\url}[1]{#1}
\csname url@samestyle\endcsname
\providecommand{\newblock}{\relax}
\providecommand{\bibinfo}[2]{#2}
\providecommand{\BIBentrySTDinterwordspacing}{\spaceskip=0pt\relax}
\providecommand{\BIBentryALTinterwordstretchfactor}{4}
\providecommand{\BIBentryALTinterwordspacing}{\spaceskip=\fontdimen2\font plus
\BIBentryALTinterwordstretchfactor\fontdimen3\font minus
  \fontdimen4\font\relax}
\providecommand{\BIBforeignlanguage}[2]{{%
\expandafter\ifx\csname l@#1\endcsname\relax
\typeout{** WARNING: IEEEtran.bst: No hyphenation pattern has been}%
\typeout{** loaded for the language `#1'. Using the pattern for}%
\typeout{** the default language instead.}%
\else
\language=\csname l@#1\endcsname
\fi
#2}}
\providecommand{\BIBdecl}{\relax}
\BIBdecl

\bibitem{ljung1999system}
L.~Ljung, \emph{System Identification: Theory for the User}.\hskip 1em plus
  0.5em minus 0.4em\relax Upper Saddle River, NJ: Prentice-Hall, 1999.

\bibitem{ioannou2012robust}
P.~A. Ioannou and J.~Sun, \emph{Robust Adaptive Control}.\hskip 1em plus 0.5em
  minus 0.4em\relax New York, NY: Dover, 2012.

\bibitem{sutton2018reinforcement}
R.~S. Sutton and A.~G. Barto, \emph{Reinforcement Learning: An
  Introduction}.\hskip 1em plus 0.5em minus 0.4em\relax Cambridge, MA: MIT
  Press, 2018.

\bibitem{de2019formulas}
C.~De~Persis and P.~Tesi, ``Formulas for data-driven control: Stabilization,
  optimality, and robustness,'' \emph{IEEE Transactions on Automatic Control},
  vol.~65, no.~3, pp. 909--924, 2020.

\bibitem{willems2005note}
J.~C. Willems, P.~Rapisarda, I.~Markovsky, and B.~L. De~Moor, ``A note on
  persistency of excitation,'' \emph{Systems \& Control Letters}, vol.~54,
  no.~4, pp. 325--329, 2005.

\bibitem{van2020data}
H.~J. Van~Waarde, J.~Eising, H.~L. Trentelman, and M.~K. Camlibel, ``Data
  informativity: A new perspective on data-driven analysis and control,''
  \emph{IEEE Transactions on Automatic Control}, vol.~65, no.~11, pp.
  4753--4768, 2020.

\bibitem{van2020noisy}
H.~J. van Waarde, M.~K. Camlibel, and M.~Mesbahi, ``From noisy data to feedback
  controllers: Nonconservative design via a matrix {S}-lemma,'' \emph{IEEE
  Transactions on Automatic Control}, vol.~67, no.~1, pp. 162--175, 2020.

\bibitem{bisoffi2022data}
A.~Bisoffi, C.~De~Persis, and P.~Tesi, ``Data-driven control via {P}etersen’s
  lemma,'' \emph{Automatica}, vol. 145, p. 110537, 2022.

\bibitem{dorfler2023certainty}
F.~D{\"o}rfler, P.~Tesi, and C.~De~Persis, ``On the certainty-equivalence
  approach to direct data-driven {LQR} design,'' \emph{IEEE Transactions on
  Automatic Control}, vol.~68, no.~12, pp. 7989--7996, 2023.

\bibitem{nortmann2023direct}
B.~Nortmann and T.~Mylvaganam, ``Direct data-driven control of linear
  time-varying systems,'' \emph{IEEE Transactions on Automatic Control},
  vol.~68, no.~8, pp. 4888--4895, 2023.

\bibitem{li2024controller}
\BIBentryALTinterwordspacing
L.~Li, A.~Bisoffi, C.~D. Persis, and N.~Monshizadeh, ``Controller synthesis
  from noisy-input noisy-output data,'' 2024. [Online]. Available:
  \url{https://arxiv.org/abs/2402.02588}
\BIBentrySTDinterwordspacing

\bibitem{alsalti2025notes}
M.~Alsalti, V.~G. Lopez, and M.~A. M{\"u}ller, ``Notes on data-driven
  output-feedback control of linear {MIMO} systems,'' \emph{IEEE Transactions
  on Automatic Control}, 2025.

\bibitem{schmitz2024continuous}
P.~Schmitz, T.~Faulwasser, P.~Rapisarda, and K.~Worthmann, ``A continuous-time
  fundamental lemma and its application in data-driven optimal control,''
  \emph{Systems \& Control Letters}, vol. 194, p. 105950, 2024.

\bibitem{lopez2024input}
V.~G. Lopez, M.~A. M{\"u}ller, and P.~Rapisarda, ``An input-output
  continuous-time version of {W}illems’ lemma,'' \emph{IEEE Control Systems
  Letters}, 2024.

\bibitem{berberich2021data}
J.~Berberich, S.~Wildhagen, M.~Hertneck, and F.~Allg{\"o}wer, ``Data-driven
  analysis and control of continuous-time systems under aperiodic sampling,''
  \emph{IFAC-PapersOnLine}, vol.~54, no.~7, pp. 210--215, 2021.

\bibitem{bianchi2025data}
M.~Bianchi, S.~Grammatico, and J.~Cort{\'e}s, ``Data-driven stabilization of
  switched and constrained linear systems,'' \emph{Automatica}, vol. 171, p.
  111974, 2025.

\bibitem{guo2021data}
M.~Guo, C.~De~Persis, and P.~Tesi, ``Data-driven stabilization of nonlinear
  polynomial systems with noisy data,'' \emph{IEEE Transactions on Automatic
  Control}, vol.~67, no.~8, pp. 4210--4217, 2021.

\bibitem{eising2024sampling}
J.~Eising and J.~Cortes, ``When sampling works in data-driven control:
  Informativity for stabilization in continuous time,'' \emph{IEEE Transactions
  on Automatic Control}, 2024.

\bibitem{hu2025data}
Z.~Hu, C.~De~Persis, J.~W. Simpson-Porco, and P.~Tesi, ``Data-driven harmonic
  output regulation of a class of nonlinear systems,'' \emph{Systems \& Control
  Letters}, vol. 200, p. 106079, 2025.

\bibitem{isidori2003robust}
A.~Isidori, L.~Marconi, and A.~Serrani, \emph{Robust Autonomous Guidance: An
  Internal Model Approach}.\hskip 1em plus 0.5em minus 0.4em\relax London:
  Springer-Verlag, 2003.

\bibitem{isidori2017lectures}
A.~Isidori, \emph{Lectures in Feedback Design for Multivariable Systems}.\hskip
  1em plus 0.5em minus 0.4em\relax Switzerland: Springer International
  Publishing, 2017.

\bibitem{de2023event}
C.~De~Persis, R.~Postoyan, and P.~Tesi, ``Event-triggered control from data,''
  \emph{IEEE Transactions on Automatic Control}, vol.~69, no.~6, pp.
  3780--3795, 2023.

\bibitem{rapisarda2023orthogonal}
P.~Rapisarda, H.~J. van Waarde, and M.~{\c{C}}amlibel, ``Orthogonal polynomial
  bases for data-driven analysis and control of continuous-time systems,''
  \emph{IEEE Transactions on Automatic Control}, vol.~69, no.~7, pp.
  4307--4319, 2023.

\bibitem{ohta2024sampling}
Y.~Ohta and P.~Rapisarda, ``A sampling linear functional framework for
  data-driven analysis and control of continuous-time systems,'' in \emph{2024
  IEEE 63rd Conference on Decision and Control}, 2024, pp. 357--362.

\bibitem{bosso2024derivative}
\BIBentryALTinterwordspacing
A.~Bosso, M.~Borghesi, A.~Iannelli, G.~Notarstefano, and A.~R. Teel,
  ``Derivative-free data-driven control of continuous-time linear
  time-invariant systems,'' 2024. [Online]. Available:
  \url{https://arxiv.org/abs/2410.24167}
\BIBentrySTDinterwordspacing

\bibitem{anderson1974adaptive}
B.~Anderson, ``Adaptive identification of multiple-input multiple-output
  plants,'' in \emph{1974 IEEE Conference on Decision and Control including the
  13th Symposium on Adaptive Processes}.\hskip 1em plus 0.5em minus 0.4em\relax
  IEEE, 1974, pp. 273--281.

\bibitem{kreisselmeier1977adaptive}
G.~Kreisselmeier, ``Adaptive observers with exponential rate of convergence,''
  \emph{IEEE Transactions on Automatic Control}, vol.~22, no.~1, pp. 2--8,
  1977.

\bibitem{narendra1980stable}
K.~Narendra, Y.-H. Lin, and L.~Valavani, ``Stable adaptive controller design,
  part {II}: Proof of stability,'' \emph{IEEE Transactions on Automatic
  Control}, vol.~25, no.~3, pp. 440--448, 1980.

\bibitem{rizvi2019reinforcement}
S.~A.~A. Rizvi and Z.~Lin, ``Reinforcement learning-based linear quadratic
  regulation of continuous-time systems using dynamic output feedback,''
  \emph{IEEE Transactions on Cybernetics}, vol.~50, no.~11, pp. 4670--4679,
  2019.

\bibitem{narendra1989stable}
K.~S. Narendra and A.~M. Annaswamy, \emph{Stable Adaptive Systems}.\hskip 1em
  plus 0.5em minus 0.4em\relax Englewood Cliffs, NJ: Prentice-Hall, 1989.

\bibitem{andrieu2006existence}
V.~Andrieu and L.~Praly, ``On the existence of a
  {K}azantzis--{K}ravaris/{L}uenberger observer,'' \emph{SIAM Journal on
  Control and Optimization}, vol.~45, no.~2, pp. 432--456, 2006.

\bibitem{bernard2022observer}
P.~Bernard, V.~Andrieu, and D.~Astolfi, ``Observer design for continuous-time
  dynamical systems,'' \emph{Annual Reviews in Control}, vol.~53, pp. 224--248,
  2022.

\bibitem{luenberger1964observing}
D.~G. Luenberger, ``Observing the state of a linear system,'' \emph{IEEE
  Transactions on Military Electronics}, vol.~8, no.~2, pp. 74--80, 1964.

\bibitem{antsaklis1997linear}
P.~J. Antsaklis and A.~N. Michel, \emph{Linear Systems}.\hskip 1em plus 0.5em
  minus 0.4em\relax Boston, MA: Birkh{\"a}user, 2006.

\bibitem{davison1976robust}
E.~Davison, ``The robust control of a servomechanism problem for linear
  time-invariant multivariable systems,'' \emph{IEEE Transactions on Automatic
  Control}, vol.~21, no.~1, pp. 25--34, 1976.

\bibitem{Lofberg2004}
J.~L{\"{o}}fberg, ``{YALMIP} : A toolbox for modeling and optimization in
  {MATLAB},'' in \emph{In Proceedings of the CACSD Conference}, Taipei, Taiwan,
  2004.

\bibitem{mosek}
\BIBentryALTinterwordspacing
M.~ApS, \emph{The MOSEK optimization toolbox for MATLAB manual. Version 10.2},
  2024. [Online]. Available:
  \url{https://docs.mosek.com/10.2/toolbox/index.html}
\BIBentrySTDinterwordspacing

\bibitem{walsh2001scheduling}
G.~C. Walsh and H.~Ye, ``Scheduling of networked control systems,'' \emph{IEEE
  Control Systems Magazine}, vol.~21, no.~1, pp. 57--65, 2001.

\bibitem{pyrkin2019adaptive}
A.~Pyrkin and A.~Isidori, ``Adaptive output regulation of right-invertible
  {MIMO} {LTI} systems, with application to vessel motion control,''
  \emph{European Journal of Control}, vol.~46, pp. 63--79, 2019.

\bibitem{denham1974canonical}
M.~Denham, ``Canonical forms for the identification of multivariable linear
  systems,'' \emph{IEEE Transactions on Automatic Control}, vol.~19, no.~6, pp.
  646--656, 1974.

\bibitem{serrani2000semiglobal}
A.~Serrani, A.~Isidori, and L.~Marconi, ``Semiglobal robust output regulation
  of minimum-phase nonlinear systems,'' \emph{International Journal of Robust
  and Nonlinear Control: IFAC-Affiliated Journal}, vol.~10, no.~5, pp.
  379--396, 2000.

\bibitem{gantmakher2000theory}
F.~R. Gantmakher, \emph{The Theory of Matrices}.\hskip 1em plus 0.5em minus
  0.4em\relax New York, NY: Chelsea Publishing Company, 1960, vol.~1.

\end{thebibliography}

\newpage

\begin{IEEEbiography}[{\includegraphics[width=1in,height=1.25in,clip,keepaspectratio]{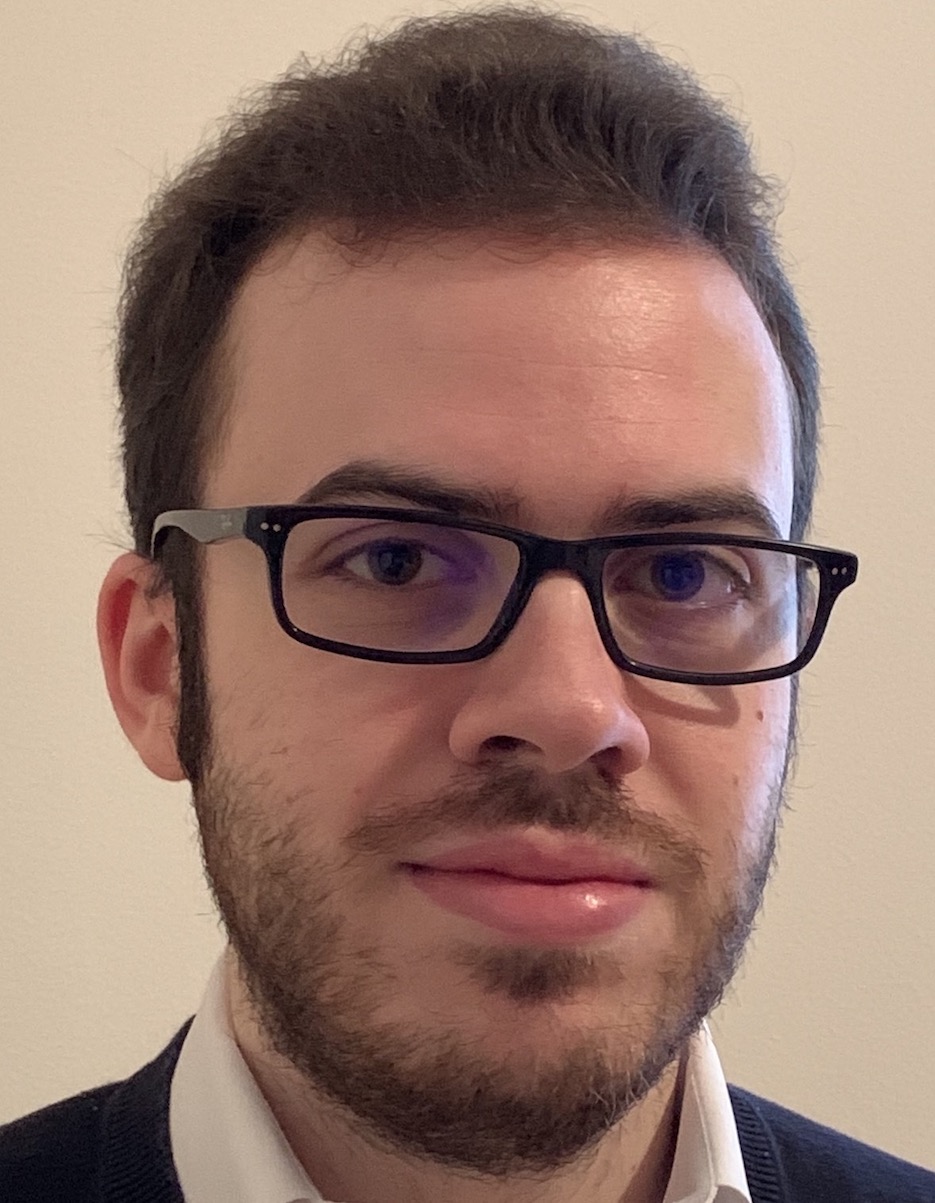}}]{Alessandro Bosso} (Member, IEEE),
    received the Ph.D. degree in Automatic Control from the University of Bologna, Bologna, Italy, in 2020.
    Currently, he is a Tenure-Track Researcher at the Department of Electrical, Electronic, and Information Engineering (DEI), University of Bologna.
    He has been a visiting scholar at The Ohio State University and at the University of California, Santa Barbara.
    His research interests include nonlinear adaptive control, hybrid dynamical systems, and the control of mechatronic systems.
    He is the recipient of a Marie Sk{\l}odowska-Curie Postdoctoral Fellowship.
\end{IEEEbiography}

\begin{IEEEbiography}[{\includegraphics[scale=2, width=1in,height=1.25in,clip,keepaspectratio]{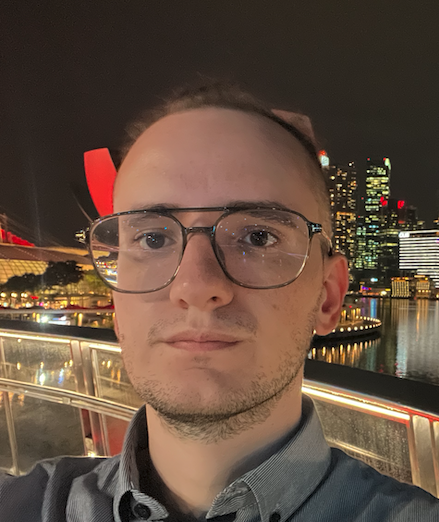}}]{Marco Borghesi} (Student Member, IEEE),
    received the master's degree in Automation Engineering from the University of Bologna in 2021 and the Ph.D. degree in Biomedical, Electrical, and Systems Engineering from the same institution in 2025.
	Currently, he is a postdoc at CASY, Department of Electrical, Electronic, and Information Engineering, University of Bologna. 
	He was a visiting scholar at the University of Stuttgart in 2024.
	His research interests include adaptive control, system theory, optimal control, and reinforcement learning.
\end{IEEEbiography}

\begin{IEEEbiography}[{\includegraphics[width=1in,height=1.25in,clip,keepaspectratio]{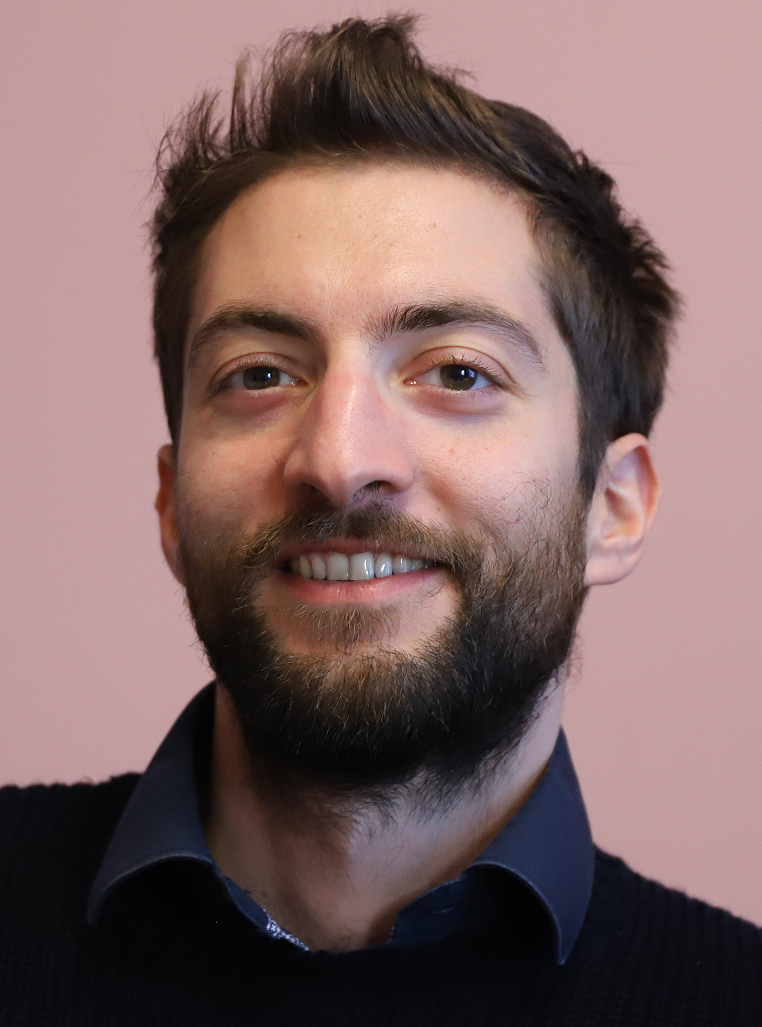}}] {Andrea Iannelli} (Member, IEEE)
    is an Assistant Professor in the Institute for Systems Theory and Automatic Control at the University of Stuttgart. He completed his B.Sc. and M.Sc. degrees in Aerospace Engineering at the University of Pisa and received his PhD from the University of Bristol. He was also a postdoctoral researcher in the Automatic Control Laboratory at ETH Zurich. His main research interests are centered around robust and adaptive control, uncertainty quantification, and sequential decision-making.  He serves the community as Associated Editor for the International Journal of Robust and Nonlinear Control and as IPC member of international conferences in the areas of control, optimization, and learning.
\end{IEEEbiography}

\begin{IEEEbiography}[{\includegraphics[width=1in,height=1.25in,clip,keepaspectratio]{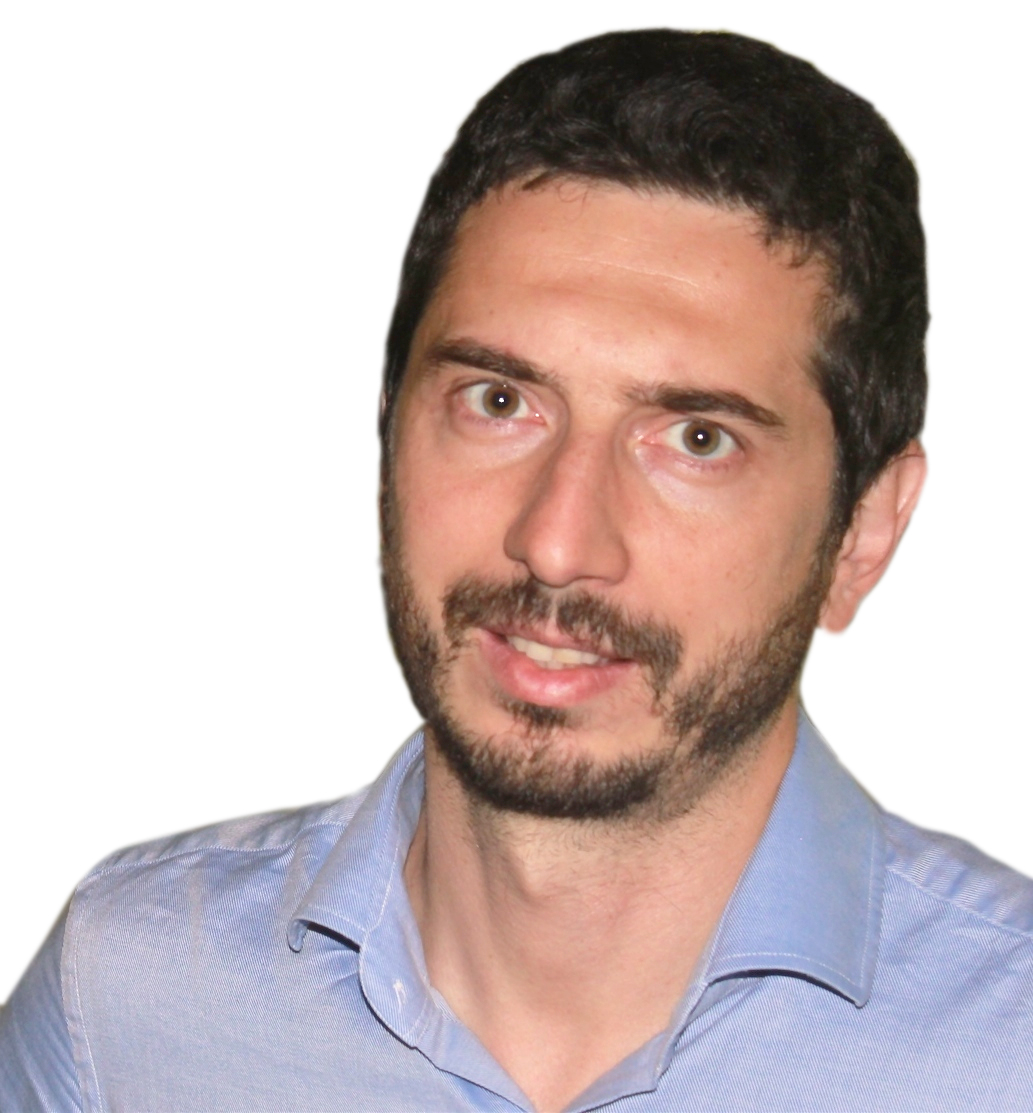}}]{Giuseppe Notarstefano} (Member, IEEE),
    is a Professor at the Department of Electrical, Electronic, and Information Engineering G. Marconi at Alma Mater Studiorum Universit\`a di Bologna. 
    He was Associate Professor (June ‘16 – June ‘18) and previously Assistant Professor, Ricercatore, (February ‘07 – June ‘16) at the Universit\`a del Salento, Lecce, Italy. 
    He received the Laurea degree “summa cum laude” in Electronics Engineering from the Universit\`a di Pisa in 2003 and the Ph.D. degree in Automation and Operation Research from the Universit\`a di Padova in 2007. 
    He has been visiting scholar at the University of Stuttgart, University of California Santa Barbara, and University of Colorado Boulder. 
	His research interests include distributed optimization, cooperative control in complex networks, applied nonlinear optimal control, and trajectory optimization and maneuvering of aerial and car vehicles. 
	He serves as an Associate Editor for IEEE Transactions on Automatic control, IEEE Transactions on Control Systems Technology and IEEE Control Systems Letters. 
	He is also part of the Conference Editorial Board of IEEE Control Systems Society and EUCA. 
	He is recipient of an ERC Starting Grant.
\end{IEEEbiography}

\begin{IEEEbiography}[{\includegraphics[width=1in,height=1.25in,clip,keepaspectratio]{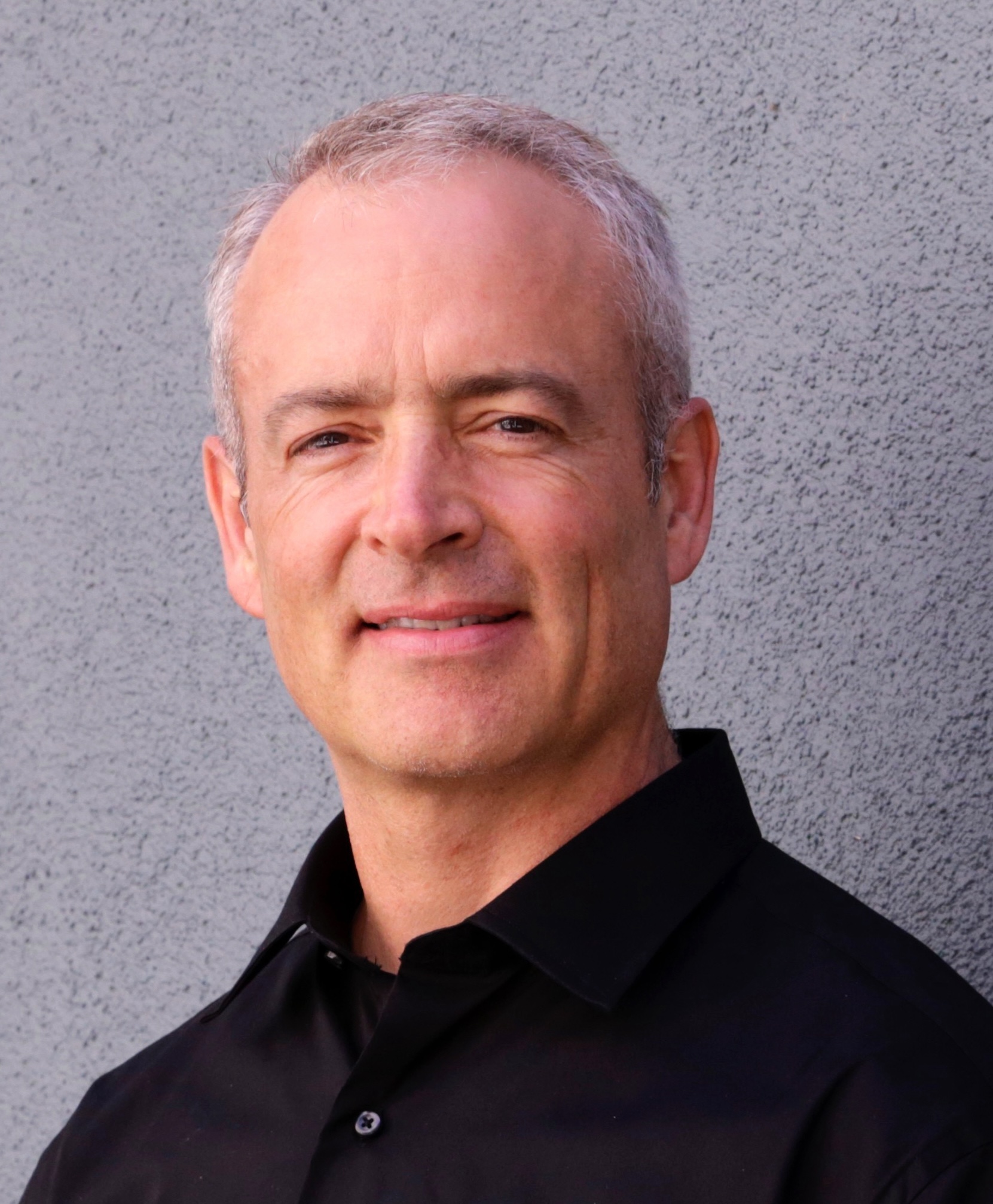}}]{Andrew R. Teel} (Fellow, IEEE),
    received his A.B. degree in Engineering Sciences from Dartmouth College in Hanover, New Hampshire, in 1987, and his M.S. and Ph.D. degrees in Electrical Engineering from the University of California, Berkeley, in 1989 and 1992, respectively. After receiving his Ph.D., he was a postdoctoral fellow at the Ecole des Mines de Paris in Fontainebleau, France. In 1992 he joined the faculty of the Electrical Engineering Department at the University of Minnesota, where he was an assistant professor until 1997. Subsequently, he joined the faculty of the Electrical and Computer Engineering Department at the University of California, Santa Barbara, where he is currently a Distinguished Professor and director of the Center for Control, Dynamical systems, and Computation.  His research interests are in nonlinear and hybrid dynamical systems, with a focus on stability analysis and control design. He has received NSF Research Initiation and CAREER Awards, the 1998 IEEE Leon K. Kirchmayer Prize Paper Award, the 1998 George S. Axelby Outstanding Paper Award, and was the recipient of the first SIAM Control and Systems Theory Prize in 1998. He was the recipient of the 1999 Donald P. Eckman Award and the 2001 O. Hugo Schuck Best Paper Award, both given by the American Automatic Control Council, and also received the 2010 IEEE Control Systems Magazine Outstanding Paper Award.  In 2016, he received the Certificate of Excellent Achievements from the IFAC Technical Committee on Nonlinear Control Systems.  In 2020, he and his co-authors received the HSCC Test-of-Time Award. He is Editor-in-Chief for Automatica, and a Fellow of the IEEE and of IFAC.
\end{IEEEbiography}

\end{document}